\documentclass[reqno]{amsart}
\usepackage[top=0.8in,bottom=0.8in,left=0.8in,right=0.8in]{geometry}
\usepackage{amscd}
\usepackage{amssymb}
\usepackage{calc,pifont}
\usepackage[T1]{fontenc}
\usepackage[latin1]{inputenc}

\usepackage[arrow,curve,matrix]{xy}
\usepackage{times}
\usepackage{graphicx}
\usepackage{extarrows}
\usepackage{flafter}
\usepackage{placeins}
\usepackage{enumitem}
\usepackage[OT2,T1]{fontenc}
\DeclareSymbolFont{cyrletters}{OT2}{wncyr}{m}{n}
\DeclareMathSymbol{\Sha}{\mathalpha}{cyrletters}{"58}

\usepackage{centernot}
\usepackage{pb-diagram}
\usepackage{mathrsfs}
\usepackage[color]{showkeys}

\definecolor{refkey}{rgb}{1,1,1}
\definecolor{labelkey}{rgb}{1,1,1}
\definecolor{cite}{rgb}{0.9451,0.2706,0.4941}
\definecolor{ruri}{rgb}{0.0078,0.4022,0.8010}

\usepackage[%
bookmarks=true,bookmarksnumbered=true,%
colorlinks=true,linkcolor=ruri,citecolor=red%
 ]{hyperref}

\theoremstyle{plain}

\newtheorem{theorem}{Theorem}[section]

\newtheorem{proposition/example}[theorem]{Proposition/Example}
\newtheorem{proposition}[theorem]{Proposition}
\newtheorem{corollary}[theorem]{Corollary}

\theoremstyle{definition}
\newtheorem{definition}[theorem]{Definition}

\newtheorem{example}[theorem]{Example}

\newtheorem{conjecture/question}[theorem]{Conjecture/Question}

\newtheorem{remark/definition}[theorem]{Remark/Definition}
\newtheorem{definition/notation}[theorem]{Definition/Notation}

 \theoremstyle{remark}

\numberwithin{equation}{section}

\usepackage{marginnote,color,changepage}
\usepackage[usenames,dvipsnames]{xcolor}

\def\shownotes{\def\inline##1##2##3{ \begin{adjustwidth}{3mm}{7mm}\mbox{}\par \noindent
{\color{##1}\hspace{-1.9cm}{\large ##2}\vspace{-\baselineskip}\\##3}
\newline\end{adjustwidth}} \def\inlinewide##1##2##3{ \begin{adjustwidth}{0mm}{0cm}\mbox{}\par \noindent
{\color{##1}\hspace{-1.6cm}{\large ##2}\vspace{-\baselineskip}\\##3}
\newline\end{adjustwidth}}  \def\marg##1##2##3{\marginnote{\color{##1}{\large ##2}\\{\small ##3}}[-.8cm]}}

\shownotes

\begin{document}
\title {Symplectic Group and Heisenberg Group in $\emph{p}$-adic Quantum Mechanics}
\author{Sen Hu \& Zhi Hu}
\address{School of Mathematics, University of Science and
Technology of China\\Hefei, 230026, China}
 \email{{\tt
shu@ustc.edu.cn, halfash@mail.ustc.edu.cn}}
\maketitle
\begin{abstract}
This paper  treats mathematically some problems in $p$-adic quantum mechanics. We first deal with $p$-adic symplectic group corresponding to the symmetry on the classical phase space. By the filtrations of isotropic subspaces and almost self-dual lattices in the $p$-adic symplectic vector space, we explicitly give the expressions of parabolic subgroups, maximal compact subgroups and corresponding Iwasawa decompositions of some symplectic groups. For a triple of Lagrangian subspaces, we associated it with a quadratic form whose Hasse invariant is  calculated. Next we study the various equivalent realizations of unique irreducible and admissible representation of  $p$-adic Heisenberg group. For the Schr\"{o}dinger representation, we can define Weyl operator and its kernel function, while for the induced representations from the characters of maximal abelian subgroups of Heisenberg group generated by the isotropic subspaces or self-dual lattice in the $p$-adic symplectic vector space, we calculate the  Maslov index defined  via the  intertwining operators corresponding to the representation transformation operators in quantum mechanics.
\end{abstract}
\tableofcontents
\section{Introduction}
It seems that there is no  prior reason for opposing the fascinating idea that  at the very small (Planck) scale the geometry of the
spacetime should be non-Archimedean. There should be quantum
fluctuations not only of topology and geometry but even of the number field.
Therefore, it was suggested the following number field invariance principle:
Fundamental physical laws should be invariant under the change of the number field\cite{v}. One could start from the ring of integers or global fields,
then the local fields should appear through a mechanism
of number field symmetry breaking, similar to the Higgs mechanism\cite{vk}.

Physics defined over $p$-adic number field has been studied in various different contexts.
In $p$-adic open string theory, the interior of the $p$-adic worldsheet is regarded as  a simple
discrete lattice or tree which can be easily embedded in a real spacetime, even though its boundary
is identified with $\mathbb{Q}_p$. A lattice action on this tree can be written down in direct analogy to that in the ordinary  string theory. The corresponding
scattering amplitudes have obtained by Brekke et al.\cite{bfw}. The tachyon effective action which  reproduces correctly all the tree level amplitudes involving the tachyon provides a non-local field theory\cite{fg}. The amplitudes for the closed $p$-adic string can be obtained by  replacing the
integrals over $\mathbb{C}$ in the usual closed string theory  with integrals over a $p$-adic quadratic extension\cite{fo}. The investigation of $p$-adic quantum mechanics began in the pioneer  paper \cite{vv} of Vladimirov and Volovich. The authors developed the corresponding
formalism of $p$-adic quantum mechanics with complex wave functions in the framework of Weyl representation of  Heisenberg group. Khrennikov
 constructed  $p$-adic
quantum mechanics with $p$-adic valued wave functions later on\cite{kh}. Since then, Volovich and collaborators have conducted some exploratory research on  $p$-adic and adelic physics\cite{vk}.

This paper will treat mathematically some problems in $p$-adic quantum mechanics. In Sect. 2, we recall briefly some basic $p$-adic analysis. Next we deal with $p$-adic symplectic group corresponding to the symmetry on the classical phase space. By the filtrations of isotropic subspaces and almost self-dual lattices in the $p$-adic symplectic vector space, we explicitly give the expressions of parabolic subgroups, maximal compact subgroups and corresponding Iwasawa decompositions of some symplectic groups. For a triple of Lagrangian subspaces, we associated it with a quadratic form whose Hasse invariant is  calculated. In the last section, we study the various equivalent realizations of unique irreducible and admissible representation of  $p$-adic Heisenberg group. For the Schr\"{o}dinger representation, we can define Weyl operator and its kernel function, while for the induced representations from the characters of maximal abelian subgroups of Heisenberg group generated by the isotropic subspaces or self-dual lattice in the $p$-adic symplectic vector space, we calculate the  Maslov index defined  via the  intertwining operators corresponding to the representation transformation operators in quantum mechanics.
\section{Brief Review of $p$-adic Analysis}
The details of $p$-adic analysis can be found in \cite{k,vvz,ak}. We only recall some basic materials used in this paper.
Let $p$ be a prime positive integer number, the
field $\mathbb{Q}_p$ of $p$-adic numbers is the completion of the rational number field $\mathbb{Q}$ with respect to the $p$-adic norm $|\cdot|_p$.  Any nonzero element $x\in\mathbb{Q}_p$ can be written uniquely as $x=p^k\sum_{n\geq0} x_np^n$ where $x_n\in\{0,\cdots,p-1\},x_0\neq0$ and $k\in \mathbb{Z}$, and the $p$-adic norm of $x$ is given by $|x|_p=\frac{1}{p^k}$.  The subset $\mathbb{Z}_p$ of $\mathbb{Q}_p$ defined as $\mathbb{Z}_p=\{x\in\mathbb{Q}_p:|x|_p\leq1\}=\{x=\sum_{n\geq0}x_np^n:x_n\in\{0,\cdots,p-1\}\}=\underleftarrow{\lim}_{n}\mathbb{Z}/p^n\mathbb{Z}$ is an integral  domain called the ring  of $p$-adic integers. The ring $\mathbb{Z}_p$  is a principal ideal domain, more precisely, its ideals
are the principal ideals $\{0\}$ and $p^k\mathbb{Z}_p (k\in\mathbb{N})$. In particular, $\mathfrak{p}=p\mathbb{Z}_p$ is the unique maximal ideal of $\mathbb{Z}_p$, the corresponding residue field is $\mathbb{F}_p=\mathbb{Z}_p/p\mathbb{Z}_p$.
 $\mathbb{Q}_p$ is also defined as the fraction field of $\mathbb{Z}_p$, i.e. $\mathbb{Q}_p=\mathbb{Z}_p[\frac{1}{p}]$.  The group multiplicative $\mathbb{Z}_p^\times$ of invertible elements in the ring $\mathbb{Z}_p$ consists of the
$p$-adic integers with unit norm, namely $\mathbb{Z}_p^\times=\{x\in\mathbb{Z}_p: |x|_p=1\}=\{x=\sum_{n\geq0}x_np^n:x_n\in\{0,\cdots,p-1\},x_0\neq0\}$, hence the subset consisting of nonzero $p$-adic number is given by
$\mathbb{Q}_p^\times=\coprod_{m\in\mathbb{Z}}p^m\mathbb{Z}_p^\times$.
For any $p$-adic number $x=\sum_{n\geq 0}x_np^{n+k}$ one defines the fractional part of $x$ as $[x]=\sum_{n=0}^{-k-1}x_np^{n+k}\in\mathbb{Z}[\frac{1}{p}]$, and for any two $p$-adic numbers $x$ and $y$, the difference $[x+y]-[x]-[y]\in \mathbb{Z}_p\cap\mathbb{Z}[\frac{1}{p}]=\mathbb{Z}$.
Call $B(a;n)=\{x\in\mathbb{Q}_p:|x-a|_p\leq p^{-n}\}$ a $p$-adic ball with center $a$. Note that every point in $B(a;n)$ is a center.  The $p$-adic balls are both open and closed sets since $B(a;n)=\{x\in\mathbb{Q}_p:|x-a|_p< p^{-n+1}\}$, and they are disconnected sets since $B(a;n)=\coprod_{x=0}^{p-1}B(a+p^nx;n+1)$.
Moreover one can show that $\mathbb{Q}_p=\cup_nB(a;n)$ is a locally compact and totally disconnected  topological field.

 An additive character $\chi$ on $\mathbb{Q}_p$ is a homomorphism $\chi:\mathbb{Q}_p\rightarrow\mathbb{C}^*$ with the property $\chi(x+y)=\chi(x)\chi(y)$. The group of additive characters of the field $\mathbb{Q}_p$ is isomorphic
to its additive group $\mathbb{Q}_p$, where the isomorphism is given by the mapping $u\mapsto\chi(ux)=\chi_u(x):=\exp(2\pi i[ux])$.  Being locally compact, $\mathbb{Q}_p$ has a real-valued Haar measure, i.e., a translation invariant measure
$dx$ with the property $d(x + a) = dx$. If the Haar measure is  normalized  so that for the compact subring $\mathbb{Z}_p$ it satisfies $\int_{\mathbb{Z}_p}dx=1$, then $dx$ is unique. For any $a\in \mathbb{Q}_p^\times$ we ave $d(ax)=|a|_pdx$. Let $\mathbb{Q}_p^n$ denote the space of product of $n$-copies of  $\mathbb{Q}_p$. The $p$-adic norm of $x=(x_1,\cdots,x_n)\in\mathbb{Q}_p^n$ is given by $|x|_p=\max_{1\leq j\leq n}|x_j|_p$, which is also a non-Archimedean norm. The Haar measure $dx$ on $\mathbb{Q}_p$ can be extended to a translation invariant measure $d^nx = dx^1 \cdots dx^n$ on $\mathbb{Q}_p^n$
in the standard way, which has the properties: $d^n(x+a)=d^nx$ ($a\in\mathbb{Q}_p^n$), $d^n(Ax)=|\det A|_pd^nx$ where $A : \mathbb{Q}_p^n
\rightarrow\mathbb{Q}_p^n$
is a linear isomorphism such that $\det A \neq 0$. Let $K$ be a measurable subset in  $\mathbb{Q}_p^n$, and $L^\alpha(K)$ ($\alpha\geq1$) be a set of all measurable
functions $f:K\rightarrow\mathbb{C}$,  such that $\int_K|f(x)|^\alpha d^nx<\infty$. A function $f\in L^1(\mathbb{Q}_p^n)$ is called integrable if there exists $\lim_{N\rightarrow\infty}\int_{(B(0;-N))^n}f(x)d^nx$. We call this limit an improper integral and denote it as
$\int_{\mathbb{Q}_p^n}f(x)d^nx$. In particular, $\int_{\mathbb{Q}_p}f(x)dx=\sum_{\nu=-\infty}^\infty\int_{|x|_p=p^\nu}f(x)dx$.
We present a general formula of change of variables in integrals.
If $x(y)$ is an analytic diffeomorphism of a clopen (close and open)
set $K_1\subset\mathbb{Q}_p$ onto a clopen set $K\subset\mathbb{Q}_p$, and $x^\prime(y)\neq0, y\in K_1$, then for any $f\in L^1(K)$ we have
$\int_Kf(x)dx=\int_{K_1}f(x(y))|x^\prime(y)|_pdy$.

The following integrals will be  very useful.
\begin{itemize}
  \item  $\int_{|x|_p=p^\nu}dx=p^\nu(1-\frac{1}{p})$.
  \item $\int_{|x|_p=p^\nu}\chi(\xi x)dx=\left\{
                                        \begin{array}{ll}
                                          p^\nu(1-\frac{1}{p}), & \hbox{$|\xi|_p\leq p^{-\nu}$;} \\
                                          -p^{\nu-1}, & \hbox{$|\xi|_p=p^{-\nu+1}$;} \\
                                          0, & \hbox{$|\xi|_p\geq p^{-\nu+2}$.}
                                        \end{array}
                                      \right.
$

  \item $\int_{\mathbb{Q}_p}f(|x|_p)\chi(\xi x)dx=(1-\frac{1}{p})\frac{1}{|\xi|_p}\sum_{\nu\geq0}p^{-\nu}f(\frac{1}{p^{\nu}|\xi|_p})-\frac{1}{|\xi|_p}f(\frac{p}{|\xi|_p})$, $\xi\neq0$.
\item $\int_{\mathbb{Q}_p}\chi(a x^2+bx)dx=\frac{\lambda_p(a)}{\sqrt{|a|_p}}\chi(-\frac{b^2}{4a}), p\neq2$,
where for a non-zero $p$-adic number $a=p^k(a_0+a_1p+\cdots)$, $\lambda_p(a)=\left\{
                                                \begin{array}{ll}
                                                  1, & \hbox{$k$ \textrm{is even};} \\
                                                  (\frac{a_0}{p}) & \hbox{$k$ \textrm{is odd and } $k\equiv1$ (\textrm{mod}4);} \\
                                                  i(\frac{a_0}{p}), & \hbox{$k$ \textrm{is odd and } $k\equiv3$ (\textrm{mod}4),}
                                                \end{array}
                                              \right.
$ with $(\frac{a_0}{p})$ denoteing the Legendre symbol.
\item$\int_{\mathbb{Q}_2}\chi(a x^2+bx)dx=\frac{\lambda_2(a)}{\sqrt{|a|_2}}\chi(-\frac{b^2}{4a})$,
where for a non-zero 2-adic number $a=2^k(1+2a_1+2^2a_2+\cdots)$, $\lambda_2(a)=\left\{
                                                \begin{array}{ll}
                                                 1+(-1)^{a_1}i, & \hbox{$k$ \textrm{is even};} \\
                                                 (-1)^{a_1+a_2} (1+(-1)^{a_1}i) & \hbox{$k$ \textrm{is odd }.}
                                                \end{array}
                                              \right.$
\end{itemize}

The set $L^2(\mathbb{Q}_p^n)$ is the Hilbert space with the scalar product $(f,g)=\int_{\mathbb{Q}_p^n}f(x)\bar g(x)d^nx$ for $f,g\in L^2(\mathbb{Q}_p^n)$ so that $||f||_{L^2}=\sqrt{(f,f)}$. In $L^2(\mathbb{Q}_p^n)$ the Cauchy-Bunjakovsky inequality holds:
$|(f,g)|\leq||f||_{L^2}\cdot||g||_{L^2}$. The Fourier transform $f \rightarrow F[f ]$ maps
$L^2(\mathbb{Q}_p^n)$ onto $L^2(\mathbb{Q}_p^n)$ one-to-one, where $F[f](\xi)=\hat{f}(\xi)=\int_{\mathbb{Q}_p^n}f(x)\chi(\xi \cdot x)d^nx$ and
$F^{-1}[\hat f](x)=f(x)=\int_{\mathbb{Q}_p^n}\hat f(\xi)\chi(-\xi \cdot x)d^n\xi$. Moreover, the Parseval-Steklov equality holds: $(f,g)=(F[f],F[g])$, i.e. the Fourier transform is a unitary
operator in $L^2(\mathbb{Q}_p^n)$.
A complex-valued function $\psi$ defined over $\mathbb{Q}_p^n$ is called locally constant if for any $x\in \mathbb{Q}_p^n$ there exists an integer $l(x)\in \mathbb{Z}$ such that $\psi(x+x^\prime)=\psi(x)$ when $x^\prime\in (B(0;-l(x)))^n$, and $l: \mathbb{Q}_p^n\rightarrow\mathbb{Z}$ is called a characteristic function associated with $\psi$. In particular, characteristic functions themselves are locally constant. We denote by $\mathcal{E}(\mathbb{Q}_p^n)$ the space of locally constant functions over  $\mathbb{Q}_p^n$,  by $\mathcal{D}(\mathbb{Q}_p^n)$ the space  consisting of all compactly supported functions belong to $\mathcal{E}(\mathbb{Q}_p^n)$, and by $\mathcal{D}^*(\mathbb{Q}_p^n)$ the set of all linear functionals on $\mathcal{D}(\mathbb{Q}_p^n)$. Then $\mathcal{D}(\mathbb{Q}_p^n)$ is dense in $L^2(\mathbb{Q}_p^n)$. If $\psi\in\mathcal{D}(\mathbb{Q}_p^n) $, we have constant characteristic functions, i.e. there exists $l\in \mathbb{Z}$ such that $\psi(x+x^\prime)=\psi(x)$ for any $x\in\mathbb{Q}_p^n,  x^\prime\in(B(0;-l))^n$.

\section{$p$-adic Symplectic Group}
Consider  a $2n$-dimensional vector space (classical phase space) $V$  over $\mathbb{Q}_p$, endowed with a non-degenerate alternating bilinear form $J$.  Let $Gr(n,V)$ be the Grassmannian variety of $n$-dimensional subspaces of $V$, and $Lag(n,V)$ be the subvariety of  $Gr(n,V)$ whose points are Lagrangian (i.e. maximal isotropic ) subspaces of $V$ with respect to $J$. For any point $E\in Lag(n,V)$ one define a set  $U_E=\{F\in Gr(n,V):V=E\oplus F\}$.  If  $F\in U_E$ is also in $Lag(n,V)$, once the basis $\{e_1,\cdots,e_n\}$
of $E$ is fixed, then one can choose a basis $\{f_{1},\cdots,f_{n}\}$ for $F_1$, which is determined by conditions $J(e_{i},f_{j})=\delta_{ij}$
for $i,j=1,\cdots,n$. After this choice, $(V,J)$ is identified with $\mathbb{Q}_p^{2n}$ with a standard symplectic form $J_0=\left(
                                                                                                                             \begin{array}{cc}
                                                                                                                               0 & Id_{n} \\
                                                                                                                              -Id_{n}& 0 \\
                                                                                                                             \end{array}
                                                                                                                           \right)
$.  Let $G(n)$ be the symplectic group of level $n$ defined over  $p$-adic field, then the group $G_{\mathbb{Q}_p}(n)$ of $\mathbb{Q}_p$-points of $G(n)$ is given by  $G_{\mathbb{Q}_p}(n)=Sp(((\mathbb{Q}_p)^{2n},J_0);\mathbb{Q}_p):=\{g\in GL(2n;\mathbb{Q}_p):J_0=g^TJ_0 g\}$ whose Lie algebra is denoted by $\mathfrak{g}_{\mathbb{Q}_p}(n)$.

\subsection{Isotropic subspaces and parabolic subgroups}

All maximal tori in  $G_{\mathbb{Q}_p}(n)$ are conjugate. We fix a $\mathbb{Q}_p$-split maximal torus $\mathbb{T}_{\mathbb{Q}_p}^n\simeq(\mathbb{Q}_p^\times)^n$ consisting of the diagonal matrices in $G_{\mathbb{Q}_p}(n)$. One can associate any maximal torus with a Weyl group  $\mathfrak{W}$. For $\mathbb{T}_{\mathbb{Q}_p}^n$ the action of $\mathfrak{W}$ on $\mathbb{T}_{\mathbb{Q}_p}^n$ is generated by
the following transformations\cite{t}:
\begin{align*}
\sigma_i:&\left(
                                                                                                  \begin{array}{cc}
                                                                                                    \textrm{diag}(x_1,\cdots,x_i,x_{i+1},\cdots,x_n)& 0 \\
                                                                                                    0 &  \textrm{diag}(x_1^{-1},\cdots,x_i^{-1},x_{i+1}^{-1},\cdots,x_n^{-1})\\
                                                                                             \end{array}
                                                                                                \right)\\ \ \ \ \ \ \ \ \ \ \ \ \ \ \ \ \  &\mapsto\left(
                                                                                                  \begin{array}{cc}
                                                                                                    \textrm{diag}(x_1,\cdots,x_{i+1},x_i,\cdots,x_n)& 0 \\
                                                                                                    0 &  \textrm{diag}(x_1^{-1},\cdots,x_{i+1}^{-1},x_i^{-1},\cdots,x_n^{-1})\\
                                                                                             \end{array}
                                                                                                \right), (1\leq i\leq n-1)\ \ \ \ \ \ \ \ \ \ \ \ \ \ \ \ \ \ \ \ \ \ \ \ \ \ \ \ \ \ \ \ \ \\
\sigma_n:&\left(
                                                                                                  \begin{array}{cc}
                                                                                                    \textrm{diag}(x_1,\cdots,x_{n-1},x_n)& 0 \\
                                                                                                    0 &  \textrm{diag}(x_1^{-1},\cdots,x_{n-1}^{-1},x_n^{-1})\\
                                                                                             \end{array}
                                                                                                \right)\\ \ \ \ \ \ \ \ \ \ \ \ \ \ \ \ \  &\mapsto\left(
                                                                                                  \begin{array}{cc}
                                                                                                    \textrm{diag}(x_1,\cdots,x_{n-1},x_n^{-1})& 0 \\
                                                                                                    0 &  \textrm{diag}(x_1^{-1},\cdots,x_{n-1}^{-1},x_n)\\
                                                                                             \end{array}
                                                                                                \right).
\end{align*}
Therefore $\mathfrak{W}$ has $2^nn!$ elements.
The simple roots
relative to $\mathbb{T}_{\mathbb{Q}_p}^n$ can be chosen as\cite{t}
\begin{align*}
  \alpha_i:& \left(
                                                                                                  \begin{array}{cc}
                                                                                                    \textrm{diag}(x_1,\cdots,x_i,x_{i+1},\cdots,x_n)& 0 \\
                                                                                                    0 &  \textrm{diag}(x_1^{-1},\cdots,x_i^{-1},x_{i+1}^{-1},\cdots,x_n^{-1})\\
                                                                                             \end{array}
                                                                                                \right) \mapsto x_ix_{i+1}^{-1}, (1\leq i\leq n-1)\\
\alpha_n:&\left(
                                                                                                  \begin{array}{cc}
                                                                                                    \textrm{diag}(x_1,\cdots,x_n)& 0 \\
                                                                                                    0 &  \textrm{diag}(x_1^{-1},\cdots,x_n^{-1})\\
                                                                                             \end{array}
                                                                                                \right) \mapsto x^2_n,
\end{align*}
each of which determines a unipotent subgroup $\mathbb{U}_i$, and then $\mathbb{B}=\mathbb{T}_{\mathbb{Q}_p}^n\mathbb{U}$ is a Borel subgroup  where
$\mathbb{U}$ is generated by all $\mathbb{U}_i$. Any parabolic subgroup of
$G_{\mathbb{Q}_p}(n)$ is conjugate to some standard parabolic subgroup   which is determined
by a  subset of $\{\alpha_1,\cdots,\alpha_n\}$\cite{pv}. They can be parameterized as follows. Take an ordered partition $(n_1,\cdots,n_l)$ of $k (0\leq k\leq n)$, and set
\[
 \mathbb{M}=\left\{\begin{pmatrix}
\,g_1\,{} \\
   &\ddots\,{}\\
     &&g_l\,{} \\
\,&&&A&&&&B\,{} \\
      &&&&(g_1^{-1})^T\,{} \\
   &&&&&\ddots\,{}\\
     &&&&&& (g_l^{-1})^T\,{} \\
&&&C&&& &D\,{}
    \end{pmatrix}:g_i\in GL(n_i,\mathbb{Q}_p),\left(
                                                \begin{array}{cc}
                                                  A & B \\
                                                  C& D \\
                                                \end{array}
                                              \right)
\in G_{\mathbb{Q}_p}(n-k) \right\},\]
then $\mathbb{P}=\mathbb{M}\mathbb{B}$ is a standard  parabolic subgroup\cite{t}. For the partition $k=n_1+\cdots+n_l (n_1\leq n_2\leq\cdots\leq n_l)$, we can consider the filtration $(0\subset W_0\subset\cdots\subset W_{l-1}\subset W_l\subset\cdots\subset W_{2l}=V)$,
where $W_i (i<l)$ is an $n_{i+1}$-dimensional isotropic subspace of $V$ and $W_i(i\geq l)=W_{2l-i-1}^\perp$ is the orthogonal complement in $V$ of $W_{2l-i-1}$ withe respect to $J$,
which induces a filtration on $\mathfrak{g}_{\mathbb{Q}_p}(n)$: $(0\subset\mathcal{W}_{-2l}\subset\cdots\subset\mathcal{W}_{2l}=\mathfrak{g}_{\mathbb{Q}_p}(n))$  determined by $\mathcal{W}_i(W_j)\subset W_{i+j}$.
Then the Lie algebra of corresponding parabolic subgroup $\mathbb{P}$ is exactly $\mathcal{W}_0$.
\begin{example}
Let $A_k$ be a $k$-dimensional  isotropic subspace of $V$ ($k\leq n$), then we have an increasing filtration on $V$: $W_\bullet=(0\subset W_0\subset W_1\subset W_2=V)$ where $W_0=A_k$, $W_1=A_k^\bot:=\{u\in V:J(u,v)=0 \textrm{ for any }v\in A_k\}$.
This filtration induces a filtration on $\mathfrak{g}_{\mathbb{Q}_p}(n)$: $\mathcal{W}_\bullet=(0\subset\mathcal{W}_{-2}\subset\mathcal{W}_{-1}\subset\mathcal{W}_{0}\subset\mathcal{W}_{1}\subset\mathcal{W}_{2}=\mathfrak{g}_{\mathbb{Q}_p})$
where $\mathcal{W}_i$ is determined by $\mathcal{W}_i(W_j)\subset W_{i+j}$.
More explicitly, take $A_k=\textrm{Span}_{\mathbb{Q}_p}\{e_1,\cdots,e_k\}$, $A_k^\bot=\textrm{Span}_{\mathbb{Q}_p}\{e_1,\cdots,e_n,f_{k+1},\cdots f_{n}\}$, then  the elements in $\mathcal{W}_i$ can be represented as the following matrices:
\begin{align*}
  \mathcal{W}_{-2}&=\left\{\left(
                       \begin{array}{cccc}
                         0 & 0 & u & 0 \\
                        0 & 0 & 0 & 0 \\
                        0 & 0 & 0 & 0 \\
                         0 & 0 & 0 & 0 \\
                       \end{array}
                     \right)
: u \textrm{\ is a symmetric }k\times k \textrm{ matrix}\right\}\simeq \textrm{Sym}(k;\mathbb{Q}_p),\\
\mathcal{W}_{-1}&=\left\{\left(
                       \begin{array}{cccc}
                         0 & s & u & v\\
                        0 & 0 & v^T & 0 \\
                        0 & 0 & 0 & 0 \\
                         0 & 0 & -s^T & 0 \\
                       \end{array}
                     \right)
:u \textrm{  is as that in }\mathcal{W}_{-2}, \textrm{ and }s,v\textrm{  are }k\times (n-k) \textrm{ matrices}\right\},\\
\mathcal{W}_{0}&=\left\{
\left(
                       \begin{array}{cccc}
                         p & s & u & v\\
                        0 & q & v^T & w \\
                        0 & 0 & -p^T & 0 \\
                         0 & r & -s^T & -q^T \\
                       \end{array}
                     \right)
:u,s,v \textrm{  are as those  in }\mathcal{W}_{-1},  \  p \textrm{  is a }k\times k \textrm{ matrix},\right.\\& \ \ \ \ \ \left.q \textrm{  is an }(n-k)\times (n-k)\textrm{ matrix} \textrm{ and }w,r\textrm{ are symmetric  }(n-k)\times (n-k) \textrm{ matrices}
\right\},\\
\mathcal{W}_{1}&=\left\{\left(
                       \begin{array}{cccc}
                         p & s & u & v\\
                        t & q & v^T & w \\
                        0 & m & -p^T & -t^T \\
                         m^T& r & -s^T & -q^T \\
                       \end{array}
                     \right)
:u,s,v,p,q,w,r \textrm{  are as those  in }\mathcal{W}_{0},\right.
\\& \ \ \ \ \ \ \ \ \ \ \ \ \ \ \ \ \ \ \ \ \ \ \ \ \ \ \ \ \ \ \ \ \ \ \ \left.t  \textrm{  is an }(n-k)\times k \textrm{ matrix} \textrm{ and } m \textrm{  is a }k\times(n-k) \textrm{ matrix}\right\}.
\end{align*}
We observe that $\mathcal{W}_{-2}$ and $\mathcal{W}_{-1}$ are both nilpotent subalgebras, and $\mathcal{W}_{2}/\mathcal{W}_{1}\simeq\mathcal{W}_{-2}$, $\mathcal{W}_{1}/\mathcal{W}_{0}\simeq\mathcal{W}_{-1}/\mathcal{W}_{-2}$.
The corresponding parabolic subgroup
whose Lie algebra is  $\mathcal{W}_0$ consists of all matrices  that may be written as $\xi=\xi_1\xi_2\xi_3\xi_4$ where
\begin{align*}
  \xi_1&=\left(
       \begin{array}{cccc}
        Id_{k}  & 0& 0 & 0 \\
         0&A & 0 & B \\
         0 & 0 & Id_{k} & 0 \\
         0 & C & 0 &D\\
       \end{array}
     \right)
\textrm{ for }\left(
          \begin{array}{cc}
            A & B \\
            C & D \\
          \end{array}
        \right)
\in G_{\mathbb{Q}_p}(n-k),\\
\xi_2&=\left(
       \begin{array}{cccc}
        g & 0 & 0 & 0 \\
         0 & Id_{n-k} & 0 & 0 \\
         0 & 0 & (g^{-1})^T & 0 \\
         0 & 0 & 0 &  Id_{n-k}\\
       \end{array}
     \right)
\textrm{ for } g\in GL(k,\mathbb{Q}_p),\\
\xi_3&=\left(
       \begin{array}{cccc}
         Id_{k} & E & 0 & F \\
          0& Id_{n-k} & F^T & 0 \\
        0 & 0 & Id_{k} &  0\\
         0 & 0 & -E^T & Id_{n-k}\\
       \end{array}
     \right)
\textrm{ for} \  k\times (n-k) \textrm{ matrices } E,F \textrm{ that satisfy} \  EF^T=FE^T,\\
\xi_4&=Id_{ n}+S \textrm{ with } S\in\mathcal{W}_{-2}.
\end{align*}
\end{example}

\subsection{Lagrangian subspaces and quadratic forms}
Let $\ell_1,\ell_2,\ell_3\in Lag(n,V)$, then one can define a quadratic form $Q$ on $\ell_1\oplus\ell_2\oplus\ell_3$ as follows:
\begin{align*}
  Q(z_1,z_2,z_3)=J(z_1,z_2)+J(z_2,z_3)+J(z_3,z_1)
\end{align*}
for $z_{1,2,3}\in\ell_{1,2,3}$. Assume that $Q$ is non-degenerate.
We denote by  $D(\ell_1,\ell_2,\ell_3)$ and $\mu(\ell_1,\ell_2,\ell_3)$ the determinant and  Hasse invariant of this  quadratic form $Q$ respectively. From the definition we observe  the following obvious properties of $\mu(\ell_1,\ell_2,\ell_3)$.
\begin{proposition}\begin{enumerate}
                    \item $\mu(\ell_1,\ell_2,\ell_3)$ is  $G_{\mathbb{Q}_p}(n)$-invariant.
  \item For any permutation $\mathcal{P}$ of the set $\{1, 2, 3\}$,
  \begin{align*}
   &\mu(\ell_{\mathcal{P}(1)},\ell_{\mathcal{P}(2)},\ell_{\mathcal{P}(3)})
   =(-1,-1)_p^{\textrm{sgn}(\mathcal{P})C_{3n}^2}(-1,D(\ell_1,\ell_2,\ell_3))^{\textrm{sgn}(\mathcal{P})(3n-1)}_p\mu(\ell_1,\ell_2,\ell_3)\\
=&\left\{
    \begin{array}{ll}
      \mu(\ell_1,\ell_2,\ell_3), & \hbox{$p\geq3$ \textrm{and } $n$ \textrm{is odd};} \\
     (-1,D(\ell_1,\ell_2,\ell_3))^{\textrm{sgn}(\mathcal{P})}_p\mu(\ell_1,\ell_2,\ell_3), & \hbox{$p\geq3$ \textrm{and } $n$ \textrm{is even};} \\
     (-1)^{\textrm{sgn}(\mathcal{P})C_{3n}^2}\mu(\ell_1,\ell_2,\ell_3), & \hbox{$p=2$ \textrm{and } $n$ \textrm{is odd};} \\
    (-1)^{\textrm{sgn}(\mathcal{P})C_{3n}^2}(-1,D(\ell_1,\ell_2,\ell_3))^{\textrm{sgn}(\mathcal{P})}_2\mu(\ell_1,\ell_2,\ell_3), & \hbox{$p=2$ \textrm{and } $n$ \textrm{is even}.}
    \end{array}
  \right.
     \end{align*}
  where $(\cdot, \cdot)_p:\mathbb{Q}_p^\times/(\mathbb{Q}_p^\times)^2\times \mathbb{Q}_p^\times/(\mathbb{Q}_p^\times)^2\rightarrow\{1,-1\} $ stands for the Hilbert symbol.
                   \end{enumerate}
\end{proposition}

   More generally   for an ordered Lagrangian  sequence $(\ell_1,\cdots,\ell_m)$ ($m\geq3$) on can introduce
\begin{align*}
 \mu(\ell_1,\cdots,\ell_m)
  :=\prod_{1\leq i<j<k\leq m }\mu(\ell_i,\ell_j,\ell_k).
\end{align*}

\begin{proposition}
Let $\ell_1,\ell_2,\ell_3,\ell_4\in Lag(1,V)$.
\begin{enumerate}
                 \item If $\ell_1=\ell_a,\ell_2=\ell_b,\ell_3=\ell_c,\ell_4=\ell_d$ where $\{a,b,c,d\}$ forms an  arithmetic sequence with a non-zero common difference, then $\mu(\ell_1,\ell_2,\ell_3,\ell_4)=1$.
 \item If the  non-degenerate quadratic form $Q(\ell_i,\ell_j,\ell_k)$ associated the triple $(\ell_i,\ell_j,\ell_k)$ represents zero for any $1\leq i<j<k\leq4$, then $\mu(\ell_1,\ell_2,\ell_3,\ell_4)=1$.
               \end{enumerate}
\end{proposition}
\begin{proof}
 For a 2-dimensional symplectic space $(V,J)\simeq(\mathbb{Q}_p^2,\det)$,  any Lagrangian subspace takes the form $\ell_a=\{(x,ax)\in\mathbb{Q}_p^2:x\in \mathbb{Q}_p\}$ with a parameter $a$ or $\ell_*=\{(0,x)\in\mathbb{Q}_p^2:x\in \mathbb{Q}_p\}$. For the triple $(\ell_a=\{(x,ax)\},\ell_b=\{(y,by)\},\ell_c=\{(z,cz)\})$ ($a\neq b, b\neq c, a\neq c$) we have\begin{align*}
      Q(\ell_a,\ell_b,\ell_c)&=(b-a)xy+(c-b)yz+(a-c)xz\\
      &=\frac{b-a}{4}(y-z+x)^2+\frac{c-b}{4}(z-x+y)^2+\frac{a-c}{4}(x-y+z)^2
    \end{align*}
The determinant  and Hasse invariant of $Q(\ell_a,\ell_b,\ell_c)$ read respectively $D(\ell_a,\ell_b,\ell_c)=\frac{1}{a-b}+\frac{1}{b-c}\in \mathbb{Q}_p^\times/(\mathbb{Q}_p^\times)^2$ and $\mu(\ell_a,\ell_b,\ell_c)=(b-a,c-b)_p(c-b,a-c)_p(b-a,a-c)_p$.
Similarly, for the triple $(\ell_a,\ell_b,\ell_*)$ ($a\neq b$) we have
\begin{align*}
  Q(\ell_a,\ell_b,\ell_*)&=(b-a)xy+yz-xz\\
  &=\frac{(y+z+(b-a-1)x)^2}{4}-\frac{(y-z-(b-a+1)x)^2}{4}+(b-a)x^2.
\end{align*}
The corresponding invariants are given by $D(\ell_a,\ell_b,\ell_*)=a-b$, and $\mu(\ell_a,\ell_b,\ell_*)=(b-a,b-a)_p$.

(1) We calculate
\begin{align*}
\mu(\ell_a,\ell_b,\ell_c,\ell_d)
  =&(b-a,c-b)_p(c-b,a-c)_p(a-c,b-a)_p\\
  &(c-b.d-c)_p(d-c,b-d)_p(b-d,c-b)_p\\
  &(c-a,d-c)_p(d-c,a-d)_p(a-d,c-a)_p\\
  &(b-a,d-b)_p(d-b,a-d)_p(a-d,b-a)_p\\
  =&((b-a)(d-c)(a-d)(c-b),(b-a)(d-c)(c-a)(d-b))_p\\
  =&(4,-3)_p=1.
\end{align*}

(2) By the theory of $p$-adic quadratic form\cite{s}, $Q(\ell_a,\ell_b,\ell_c)$ represents $A$ if $\mathbb{Q}_p^\times/(\mathbb{Q}_p^\times)^2\ni A\neq\frac{1}{b-a}+\frac{1}{c-b}$ or $A=\frac{1}{b-a}+\frac{1}{c-b},0$ and $(a-b,a-b)_p(c-b,c-b)_p(a-c,a-c)_p=(b-a,c-b)_p(c-b,a-c)_p(b-a,a-c)_p$, and  $Q(\ell_a,\ell_b,\ell_*)$  represents any number in $\mathbb{Q}_p$. Therefore, we have
\begin{align*}
\mu(\ell_a,\ell_b,\ell_c,\ell_d)
  =&(a-b,a-b)_p^2(c-b,c-b)_p^2(a-c,a-c)_p^2\\
  &(d-c,d-c)_p^2(d-b,d-b)_p^2(a-d,a-d)_p^2\\
  =&1,\end{align*}
and
 \begin{align*}
 &\mu(\ell_a,\ell_b,\ell_c,\ell_*)=
 \mu(\ell_*, \ell_a,\ell_b,\ell_c)=\mu(\ell_a,\ell_b,\ell_*, \ell_c)=\mu(\ell_a,\ell_*,\ell_b,\ell_c)\\
 =&(b-a,c-b)_p(c-b,a-c)_p(a-c,b-a)_p\\
 &(c-b,c-b)_p(c-a,c-a)_p(b-a,b-a)_p\\
 =&(-1,(a-b)(a-c)(c-b))_p^2=1.
\end{align*}
We complete the proof.
\end{proof}

\subsection{Latices and maximal compact subgroups}

If $L$ is a finitely generated $\mathbb{Z}_p$-submodule of $V$ containing a basis of $V$, it is called a lattice in $V$. The dual lattice $L^*$ with respect to $J$ of $L$ is defined as $L^*=\{v\in V: J(v,u)\in \mathbb{Z}_p$ for $\forall u\in L\}$. If $L=L^*$, $L$ is called a self-dual lattice.
Similarly, one can define the  1-almost self-dual lattice $L$ in $(V,J)$ if it contains a self-dual lattice and $J(u,v)\in p^{-1}\mathbb{Z}_p$ for any $u,v\in L$. By choosing a suitable symplectic basis any 1-almost self-dual lattice can be reduced
 to one of  $L_i:=\mathbb{Z}_pe_1\oplus\cdots\oplus\mathbb{Z}_pe_n\oplus\mathfrak{p}^{-1}f_1\oplus\cdots\oplus\mathfrak{p}^{-1}f_i\oplus\mathbb{Z}_pf_{i+1}\oplus\cdots\oplus\mathbb{Z}_pf_n(i=0,\cdots,n)$, whose corresponding dual latices are given by $L_i^*:=\mathfrak{p}e_1\oplus\cdots\oplus\mathfrak{p}e_i\oplus\mathbb{Z}_pe_{i+1}\oplus\cdots\oplus\mathbb{Z}_pe_n\oplus\mathbb{Z}_pf_1\oplus\cdots\oplus\mathbb{Z}_pf_n$. Then we have a flag $L_n^*\subset L_{n-1}^*\subset\cdots\subset L_1^*\subset L_0\subset L_1 \subset\cdots\subset L_{n-1}\subset L_n$ with the property that $L_i/L_{i-1}\simeq L_{i-1}^*/L_i^*\simeq(\mathbb{F}_p)^i$. The stabilizers of $L_i (\textrm{or} \ L_i^*)$ in
$G_{\mathbb{Q}_p}(n)$ are denoted by $G(L_i)($ or $G(L_i^*))$ respectively. They are all  maximal paraholic subgroups, hence are maximal compact subgroups of $G_{\mathbb{Q}_p}(n)$, and they are not conjugate to each other in $G_{\mathbb{Q}_p}(n)$\cite{j}. Each maximal compact subgroup of $G_{\mathbb{Q}_p}(n)$ is conjugate to one of $G(L_i, n)(=G(L_i^*,n))$.  The set  of all
1-almost self-dual  lattices  is given by $\Lambda_1(n)\simeq\bigsqcup_i G_{\mathbb{Q}_p}(n)/G(L_i, n)$. Recursively, one can introduce the $l$-almost self-dual lattice $L$ which contains an $(l-1)$-almost self-dual lattice and satisfies $J(u,v)\in p^{-l}\mathbb{Z}_p$ for any $u,v\in L$. An $l$-almost self-dual lattice is called the pure $l$-almost self-dual lattice  if it is not an $(l-1)$-almost self-dual lattice.

\begin{proposition}\begin{enumerate}
                     \item Let $L$ and $L^\prime$  be pure $l$-almost and $k$-almost ($k\leq l$)  self-dual lattices in a  2-dimensional symplectic space $(V,J)$  respectively, then $L^\prime$ is a sub-lattice of $L$ if and only if there exists $\Theta\in GL(2,\mathbb{Z}_p)$ with $\det\Theta= p^{l-k}$ up to a constant factor that is a unit such that $L=\Theta L^\prime$.
                     \item There is a one-to-one correspondence between  the set $\Lambda_0(2)\simeq G_{\mathbb{Q}_p}(2)/G_{\mathbb{Z}_p}(2)$ (where $G_{\mathbb{Z}_p}(n)=G(n)\cap GL(n,\mathbb{Z}_p)$) of all
self-dual lattices in two dimensional symplectic space $(V,J)$ and $\mathbb{Z}^2\ltimes((\mathbb{Q}_p)^4/\sim)$, where the equivalent relation is defined as follows: $(\mathbb{Q}_p)^4\ni(u,s,r,t)\sim (u^\prime,s^\prime,r^\prime,t^\prime)$ if and only if $u^\prime-u\in \mathbb{Z}_p$, $t^\prime-t\in \mathbb{Z}_p$, $(r^\prime-r)+t(u^\prime-u)\in \mathbb{Z}_p$ and $(s^\prime-s)+(r+r^\prime)(u^\prime-u)\in \mathbb{Z}_p$.
\item There are  one-to-one correspondences between the sets $\Lambda_1(1)$ and $(\mathbb{Z}\ltimes (\mathbb{Q}_p/\mathbb{Z}_p))\sqcup (\mathbb{Z}\ltimes (\mathbb{Q}_p/\mathbb{Z}_p))$, and between the sets $\Lambda_1(1)$ and $(\mathbb{Z}^2\ltimes((\mathbb{Q}_p)^4/\sim))\sqcup\Lambda_0(2)$, where the equivalent relation is defined as follows: $(\mathbb{Q}_p)^4\ni(u,s,r,t)\sim (u^\prime,s^\prime,r^\prime,t^\prime)$ if and only if $u^\prime-u\in \mathbb{Z}_p$, $t^\prime-t\in \mathbb{Z}_p$, $p(r^\prime-r)+t(u-u^\prime)\in \mathfrak{p}$, $p(r^\prime-r)+t^\prime(u-u^\prime)\in \mathbb{Z}_p$ and $(s^\prime-s)+(r+r^\prime)(u-u^\prime)\in\mathbb{Z}_p$.
                   \end{enumerate}

\end{proposition}

\begin{proof}(1) $L^\prime$ is a sub-lattice of $L$ if and only if there exists $\Theta\in GL(2,\mathbb{Z}_p)$ such that $L=\Theta L^\prime$. For  $\Theta=\left(
                                \begin{array}{cc}
                                  a & b\\
                                  c & d \\
                                \end{array}
                              \right)$, we can always make one of $a,b,c,d$ to be a unit via extracting  a suitable factor $p^m (m\in\mathbb{Z}_+)$,  without lost
of generality, we may assume that $a$ or $c$ is a unit. Then we  have the following decomposition
\begin{align*}
 \left(
                                \begin{array}{cc}
                                  a & b\\
                                  c & d \\
                                \end{array}
                              \right)=\left\{
                                        \begin{array}{ll}
                                         \left(
                                           \begin{array}{cc}
                                             1 & \frac{a}{c} \\
                                             0 & 1 \\
                                           \end{array}
                                         \right)\left(
                                                  \begin{array}{cc}
                                                    0 & -1 \\
                                                    1 & 0 \\
                                                  \end{array}
                                                \right)\left(
                                                                \begin{array}{cc}
                                                                  c & 0 \\
                                                                  0& \frac{\det\Theta}{c} \\
                                                                \end{array}
                                                              \right)\left(
                                                         \begin{array}{cc}
                                                           1 & \frac{d}{c} \\
                                                           0 & 1\\
                                                         \end{array}
                                                       \right),
  & \hbox{$|c|_p=1$;} \\
                                          \left(
                                            \begin{array}{cc}
                                              1 & 0 \\
                                             \frac{c}{a} & 1 \\
                                            \end{array}
                                          \right)
\left(
                                            \begin{array}{cc}
                                              a &  0\\
                                              0 & \frac{\det\Theta}{a} \\
                                            \end{array}
                                          \right)
\left(
                                            \begin{array}{cc}
                                              1 & \frac{b}{a}\\
                                              0 & 1 \\
                                            \end{array}
                                          \right)
, & \hbox{$|a|_p=1$.}
                                        \end{array}
                                      \right.
\end{align*}
thus $\Theta$ can be rewritten as the form $\Theta=A\left(
                                                      \begin{array}{cc}
                                                        \alpha & 0 \\
                                                        0& \frac{\det\Theta}{\alpha} \\
                                                      \end{array}
                                                    \right)B
$ where $A,B\in SL(2,\mathbb{Z}_p)$ and $\alpha\in\mathbb{Z}_p^\times$. Let $\{e,f\}$ and $\{e^\prime,f^\prime\}$ be bases of $L$ and
$L^\prime$ respectively, then $\{e_1,f_1\}:=\{A^{-1}e,A^{-1}f\}$ and $\{e_2,f_2\}:=\{Be^\prime,Bf^\prime\}$ are still  bases, and $e_1=\alpha e_2$, $f_1=\frac{\det\Theta}{\alpha}f_2$.  Therefore $J(e_1,f_1)=p^l\mu=\det\Theta J(e_2,f_2)=p^k\nu$ where $\mu,\nu\in\mathbb{Z}_p^\times$ due to the purity of latices, i.e. $\det\Theta=p^{l-k}\mu\nu^{-1}$.

(2) Let $L_0, L_0^\prime$ be two self-dual  latices  in $(V,J)$,  there is a symplectic basis $\{e_1,e_2,f_1,f_2\}$ and an element $g\in G_{\mathbb{Q}_p}(2)$ such that
$L_0=\mathbb{Z}_pe_1\oplus\mathbb{Z}_pe_2\oplus\mathbb{Z}_pf_1\oplus\mathbb{Z}_pf_2$ and $L_0^\prime=gL_0$. By Iwasawa decomposition, $g$ can be expressed as
$$g=\begin{pmatrix}
p^{a_1}\,{} \\
   & p^{a_2}\,{} \\
     &&p^{-a_1}\,{} \\
   &&& p^{-a_2}\,{}
\end{pmatrix}\left(
               \begin{array}{cccc}
                 1 & u & 0 & 0 \\
                 0 & 1 & 0 & 0 \\
                 0 & 0 & 1 & 0 \\
                 0 & 0 & -u & 1 \\
               \end{array}
             \right)\left(
                      \begin{array}{cccc}
                        1 & 0 & s & r \\
                        0 & 1& r &t  \\
                        0& 0 & 1& 0 \\
                        0 & 0 & 0 & 1 \\
                      \end{array}
                    \right)
h,
$$
where $a_1, a_2\in\mathbb{Z}$, $u,r,s,t\in\mathbb{Q}_p$, and
$h$ belongs to the group $G_{\mathbb{Z}_p}(2)$ of $\mathbb{Z}_p$-points in $G(2)$ which preserves the lattice $L_0$. Therefore $L_0^\prime$ take the form as
\begin{align*}
 L_0^\prime=&\mathbb{Z}_p p^{a_1}e_1\oplus\mathbb{Z}_p (p^{a_1}ue_1+p^{a_2}e_2)\\
 &\oplus\mathbb{Z}_p(p^{a_1}se_1+p^{a_1}rue_1+p^{a_2}re_2 +p^{-a_1}f_1-p^{-a_2}uf_2)\\&\oplus\mathbb{Z}_p(p^{a_1}re_1+p^{a_1}ute_1+ p^{a_2}te_2+p^{-a_2}f_2).
\end{align*}
On the other hand, if $g^\prime\in  G_{\mathbb{Q}_p}(2)$ satisfies $g^\prime g^{-1}\in G_{\mathbb{Z}_p}(2)$, then $g^\prime L_0=gL_0=L_0^\prime$. More precisely, we need

\begin{align*}
&\left(
                      \begin{array}{cccc}
                        1 & 0 & -s & -r \\
                        0 & 1& -r &-t  \\
                        0& 0 & 1& 0 \\
                        0 & 0 & 0 & 1 \\
                      \end{array}
                    \right)\left(
               \begin{array}{cccc}
                 1 & -u & 0 & 0 \\
                 0 & 1 & 0 & 0 \\
                 0 & 0 & 1 & 0 \\
                 0 & 0 & u & 1 \\
               \end{array}
             \right)\begin{pmatrix}
p^{a^\prime_1-a_1}\,{} \\
  & p^{a^\prime_2-a_2}\,{} \\
     &&p^{a_1-a^\prime_1}\,{} \\
  &&& p^{a_2-a^\prime_2}\,{}
\end{pmatrix}\\
&\times\left(
               \begin{array}{cccc}
                 1 & u^\prime & 0 & 0 \\
                 0 & 1 & 0 & 0 \\
                 0 & 0 & 1 & 0 \\
                 0 & 0 & -u^\prime & 1 \\
               \end{array}
             \right)\left(
                      \begin{array}{cccc}
                        1 & 0 & s^\prime & r^\prime \\
                        0 & 1& r^\prime &t^\prime  \\
                        0& 0 & 1& 0 \\
                        0 & 0 & 0 & 1 \\
                      \end{array}
                    \right)
\\
=&\left(
   \begin{array}{cccc}
     p^{a_1^\prime-a_1} & u^\prime p^{a_1^\prime-a_1}-up^{a_2^\prime-a_2}   &\begin{array}{c}
                                                                             (-s-ru) p^{a_1-a_1^\prime}+ru^\prime p^{a_2-a_2^\prime} \\
                                                                              +(s^\prime+u^\prime r^\prime) p^{a^\prime_1-a_1}-r^\prime up^{a_2^\prime-a_2}
                                                                             \end{array}
  &\begin{array}{c}
    -rp^{a_2-a_2^\prime}+(r^\prime+u^\prime t^\prime)p^{a_1^\prime-a_1} \\
     -ut^\prime p^{a_2^\prime-a_2}
   \end{array}
 \\
     0 & p^{a_2^\prime-a_2}  & \begin{array}{c}
                              r^\prime p^{a_2^\prime-a_2}-(r+u t)p^{a_1-a_1^\prime}\\
                                 +u^\prime t p^{a_2-a_2^\prime}
                               \end{array}
  & t^\prime p^{a_2^\prime-a_2}-t p^{a_2-a_2^\prime} \\
     0 & 0 & p^{a_1-a_1^\prime} & 0 \\
     0 & 0 &  up^{a_1-a_1^\prime}-u^\prime p^{a_2-a_2^\prime}  & p^{a_2-a_2^\prime}  \\
   \end{array}
 \right)\\
\in& G_{\mathbb{Z}_p}(2),
\end{align*}
which implies  $a_1=a^\prime_1,a_2=a^\prime_2$ and the equivalent relation.

(3) Consider the case $n=1$, the  1-almost self-dual lattice $L_1$ is given by $L_1=\mathbb{Z}_pe_1\oplus p^{-1}\mathbb{Z}_pf_1$, and  the corresponding maximal compact subgroup
$G(L_1,1)$ is given by \begin{align*}
                        G(L_1,1)&=\left\{\left(
                                     \begin{array}{cc}
                                       a & pb\\
                                       p^{-1}c & d \\
                                     \end{array}
                                   \right)
                        :a,b,c,d\in\mathbb{Z}_p,ad-bc=1\right\}\\
                        &=\left(
                           \begin{array}{cc}
                             1 & 0 \\
                            \mathfrak{p}^{-1} & 1 \\
                           \end{array}
                         \right)\left(
                                  \begin{array}{cc}
                                    \alpha & 0 \\
                                    0& \alpha^{-1} \\
                                  \end{array}
                                \right)_{\alpha\in \mathbb{Z}_p^\times}
\left(
                                  \begin{array}{cc}
                                    1 & \mathfrak{p} \\
                                    0 & 1 \\
                                  \end{array}
                                \right)\\
&\simeq SL(2,\mathbb{Z}_p),
                       \end{align*} which implies $\Lambda_1(1)\simeq\Lambda_0(1)\bigsqcup\Lambda_0(1)$, where $\Lambda_0(1)\simeq\mathbb{Z}\ltimes (\mathbb{Q}_p/\mathbb{Z}_p)$\cite{z}.

 For the case $n=2$, we have $L_1=\mathbb{Z}_pe_1\oplus\mathbb{Z}_pe_2\oplus \mathfrak{p}^{-1}f_1\oplus\mathbb{Z}_pf_2$, and then
\begin{align*}
 G(L_1,2)=&\left(
            \begin{array}{cccc}
              \mathbb{Z}_p& \mathbb{Z}_p & \mathfrak{p} &\mathbb{Z}_p \\
              \mathbb{Z}_p & \mathbb{Z}_p &\mathfrak{p}  &\mathbb{Z}_p \\
             \mathfrak{p}^{-1}  & \mathfrak{p}^{-1} & \mathbb{Z}_p & \mathfrak{p}^{-1} \\
              \mathbb{Z}_p & \mathbb{Z}_p&\mathfrak{p} & \mathbb{Z}_p \\
            \end{array}
          \right)\bigcap G_{\mathbb{Q}_p}(2)\\
          =&\left\{\left(
                     \begin{array}{cccc}
                       1 & 0 & 0 & 0\\
                       a & 1 & 0& 0 \\
                      \mathfrak{p}^{-1}& \mathbb{Z}_p & 1 & -a \\
                       \mathbb{Z}_p & \mathbb{Z}_p & 0 & 1 \\
                     \end{array}
                   \right)
         \begin{pmatrix}
\alpha\,{} \\
   & \beta\,{} \\
     &&\alpha^{-1}\,{} \\
   &&& \beta^{-1}\,{}
\end{pmatrix}\left(
                     \begin{array}{cccc}
                       1 & pb & \mathfrak{p} &  \mathfrak{p } \\
                       0& 1 &  \mathfrak{p }&\mathbb{Z}_p \\
                       0 & 0 & 1 & 0 \\
                       0 & 0 &  -pb & 1 \\
                     \end{array}
                   \right):
          a,b\in\mathbb{Z}_p; \alpha,\beta\in\mathbb{Z}_p^\times\right\}.
\end{align*}
Thereby the  corresponding Iwasawa decomposition is given by
$$g=\begin{pmatrix}
p^{a_1}\,{} \\
   & p^{a_2}\,{} \\
     &&p^{-a_1}\,{} \\
   &&& p^{-a_2}\,{}
\end{pmatrix}\left(
               \begin{array}{cccc}
                 1 & 0 & 0 & 0 \\
                 u & p & 0 & 0 \\
                 0 & 0 & 1 & -u \\
                 0 & 0 & 0 & p^{-1} \\
               \end{array}
             \right)\left(
                      \begin{array}{cccc}
                        1 & 0 & 0 & 0 \\
                        0 & 1& 0 &0  \\
                        s& pr & 1& 0 \\
                        r & t & 0 & 1 \\
                      \end{array}
                    \right) h$$
                    for $h\in G(L_1,2)$.
$L_1$ induces a lattice $L_1^\prime$ via a symplectic transformation,  which takes the form as
\begin{align*}
 L_1^\prime=&\mathbb{Z}_p( p^{a_1}e_1+p^{a_2}ue_2+p^{-a_1-1}(s-ur)f_1+p^{-a_2-1}rf_2)\\
 &\oplus\mathbb{Z}_p(p^{a_2+1}e_2+p^{-a_1}rf_1-p^{-a_1-1}utf_1+p^{-a_2-1}tf_2)\\
&\oplus\mathbb{Z}_pp^{-a_1-1}f_1\oplus\mathbb{Z}_p(-p^{-a_1-1}uf_1+p^{-a_2-1}f_2).
\end{align*}
From the same argument in (2), it follows that $g^\prime L_1=gL_1$ if and only if
\begin{align*}
 &\left(
   \begin{array}{cccc}
     p^{a_1^\prime-a_1} &
0  &0
  &0
 \\
     u^\prime p^{a_2^\prime-a_2-1}-up^{a_1^\prime-a_1-1}  & p^{a_2^\prime-a_2}  & 0
  & 0 \\
    \begin{array}{c}
                                                                             (ur-s) p^{a_1^\prime-a_1}-ru^\prime p^{a_2^\prime-a_2} \\
                                                                              +(s^\prime-u^\prime r^\prime) p^{a_1-a^\prime_1}+r^\prime up^{a_2-a_2^\prime}
                                                                             \end{array}& \begin{array}{c}
    -rp^{a_2^\prime-a_2+1}+(r^\prime p-u^\prime t^\prime)p^{a_1-a_1^\prime} \\
     +ut^\prime p^{a_2-a_2^\prime}
   \end{array} & p^{a_1-a_1^\prime} & up^{a_2-a_2^\prime}-u^\prime p^{a_1-a_1^\prime} \\
   \begin{array}{c}
                              r^\prime p^{a_2-a_2^\prime}+(ut-rp)p^{a_1^\prime-a_1-1}\\
                                 -u^\prime t p^{a_2^\prime-a_2-1}
                               \end{array} &t^\prime p^{a_2-a_2^\prime}-t p^{a_2^\prime-a_2} &  0  & p^{a_2-a_2^\prime}  \\
   \end{array}
 \right)\\
\in&Sp((\mathbb{Q}_p^4,J_1);\mathbb{Z}_p),
\end{align*}
where $J_1=\left(
             \begin{array}{cccc}
               0 & 0 & p^{-1} & 0 \\
               0 & 0 & 0& 1 \\
               -p^{-1} & 0 & 0 & 0 \\
               0 & -1 & 0 & 0 \\
             \end{array}
           \right)
$.

Similarly, for $L_2=\mathbb{Z}_pe_1\oplus\mathbb{Z}_pe_2\oplus\mathfrak{p}^{-1}f_1\oplus\mathfrak{p}^{-1}f_2$  we have
\begin{align*}
     &G(L_2,2)= \left(
            \begin{array}{cccc}
              \mathbb{Z}_p& \mathbb{Z}_p & \mathfrak{p} &\mathfrak{p} \\
              \mathbb{Z}_p & \mathbb{Z}_p &\mathfrak{p}  &\mathfrak{p} \\
              \mathfrak{p}^{-1} &\mathfrak{p}^{-1} & \mathbb{Z}_p &  \mathbb{Z}_p \\
            \mathfrak{p}^{-1}&\mathfrak{p}^{-1}&\mathbb{Z}_p & \mathbb{Z}_p \\
            \end{array}
          \right)\bigcap G_{\mathbb{Q}_p}(2)\\
           =&\left\{\left(
                     \begin{array}{cccc}
                       1 & 0 & 0 & 0\\
                       a & 1 & 0& 0 \\
                       \mathfrak{p}^{-1} & \mathfrak{p}^{-1} & 1 & -a \\
                    \mathfrak{p}^{-1} & \mathfrak{p}^{-1} & 0 & 1 \\
                     \end{array}
                   \right)
          \begin{pmatrix}
\alpha\,{} \\
   & \beta\,{} \\
     &&\alpha^{-1}\,{} \\
   &&& \beta^{-1}\,{}
\end{pmatrix}\left(
                     \begin{array}{cccc}
                       1 &  b& \mathfrak{p} &  \mathfrak{p } \\
                       0& 1 &  \mathfrak{p }&\mathfrak{p} \\
                       0 & 0 & 1 & 0 \\
                       0 & 0 &  -b & 1 \\
                     \end{array}
                   \right):a,b\in\mathbb{Z}_p; \alpha,\beta\in\mathbb{Z}_p^\times \right\}\\
                   &\bigcup\left\{\left(
                     \begin{array}{cccc}
                       1 & 0 & 0 & 0\\
                       a & 1 & 0& 0 \\
                      \mathfrak{p}^{-1}& \mathfrak{p}^{-1} & 1 & -a \\
                     \mathfrak{p}^{-1}& \mathfrak{p}^{-1} & 0 & 1 \\
                     \end{array}
                   \right)
         \begin{pmatrix}
\alpha\,{} \\
   & \beta\,{} \\
     &&\alpha^{-1}\,{} \\
   &&& \beta^{-1}\,{}
\end{pmatrix}\left(
                                       \begin{array}{cccc}
                                         0 & 1 & 0 & 0 \\
                                         1 & 0 & 0 & 0 \\
                                         0 & 0 & 0 & 1 \\
                                         0 & 0& 1 & 0 \\
                                       \end{array}
                                     \right)
          \left(
                     \begin{array}{cccc}
                       1 &  b& \mathfrak{p} &  \mathfrak{p } \\
                       0& 1 &  \mathfrak{p }&\mathfrak{p} \\
                       0 & 0 & 1 & 0 \\
                       0 & 0 &  -b & 1 \\
                     \end{array}
                   \right)
\right.\\& \ \ \ \ \ \left.:a,b\in\mathbb{Z}_p; \alpha,\beta\in\mathbb{Z}_p^\times \right\}.
\end{align*}
From these explicit expressions we can deduce the conclusions.
\end{proof}

\section{$p$-adic Heisenberg Group}
The standard quantum mechanics starts with a representation of the well-known
Heisenberg commutation relation $[\hat q,\hat p]\subset iId$ where $\hat q,\hat p$ are unbounded self-adjoint linear operator on a Hilbert space $\mathcal{H}$
and the domain of $[\hat q,\hat p]$ is a dense subset in $\mathcal{H}$ and $[\hat q,\hat p]|_{\textrm{Dom}([\hat q,\hat p])}=iId$. Let $\{E_\lambda,-\infty<\lambda<\infty\}$ be the spectrum of $\hat p$, i.e. $p=\int_{-\infty}^\infty\lambda dE_\lambda$. Define a family of unitary operators $U(t)=e^{i\hat pt}
:=\int_{-\infty}^{\infty}e^{it\lambda }dE_\lambda, t\in\mathbb{R}$. Similarly one may define $V(t)=e^{i\hat q t}$. Then the
Heisenberg commutation relation is equivalent to  the relation
\begin{equation}\label{eq:a}
  U(t)V(s)=e^{its}V(s)U(t).
\end{equation}
Another approach is to define a unitary operator $Q(z)=U(t)V(s)e^{-\frac{i}{2}ts}$ where $z=(t,s)\in\mathbb{R}^2$. Then \eqref{eq:a} is equivalent to the relation
\begin{equation}\label{eq:b}
  Q(z)Q(z^\prime)=e^{\frac{i}{2}J_0(z,z^\prime)}Q(z+z^\prime),
\end{equation}
where $J_0(z,z^\prime)=ts^\prime-t^\prime s$.  \eqref{eq:a} or  \eqref{eq:b} is called Weyl formulation of quantum mechanics where one represents Heisenberg group rather than its algebra. It provides an appropriate framework for a generalization to $p$-adic variables.
\begin{definition}
The
$p$-adic Heisenberg group $H_p(V)$ over  a $2n$-dimensional symplectic  vector space $(V,J)$ is the group extension of $V$ by a unit circle $\mathbf{S}^1$ in the field $\mathbb{C}$, i.e. the set of  pairs
$\{(u;\alpha):u\in V,\alpha\in\mathbf{S}^1\}$ with the composition law
\begin{align}\label{c}
 (u;\alpha)\cdot (v;\beta)=(u+v;\alpha\beta\chi(\frac{1}{2}J (u,v))).
\end{align}
\end{definition}
\subsection{Schr\"{o}dinger representations and  Weyl operators}

Under the Schr\"{o}dinger picture, the Heisenberg group $H(V)$  can be realized on the group $\mathcal{U}$ of unitary operators  acting on the space $\mathcal{D}(\mathbb{Q}_p^n)$ as follows. Let us choose a symplectic basis such that $(V,J)\simeq (\mathbb{Q}_p^{2n},J_0)$, and equip $\mathcal{U}$ with a weak operator topology,  then we define a continuous homomorphism called the Schr\"{o}dinger representation  $\Phi:H(V)\rightarrow\mathcal{U}$ via $\Phi(g)[\psi](\xi)=\alpha T_z[\psi](\xi)$, where
\begin{align}\label{1}
  T_z[\psi](\xi)=\psi(\xi+x)\chi(\sum_i(y_i\xi_i+\frac{1}{2}x_iy_i))=\psi(\xi+x)\prod_i\chi(y_i\xi_i+\frac{1}{2}x_iy_i)
\end{align}
for $g=(z=(x_1,\cdots,x_n,y_1,\cdots y_n);\alpha)\in H(V)$ in terms of the chosen symplectic basis of $V$ and $\xi=(\xi_1,\cdots, \xi_n)\in \mathbb{Q}_p^n$.
\begin{proposition}
The  Schr\"{o}dinger representation  is an irreducible and admissible representation.
\end{proposition}
\begin{proof}
Suppose $\mathbb{W}$ is a non-tivial $H(V)$-invariant   subspace of $\mathcal{D}(\mathbb{Q}_p^n)$. One can choose an orthonormal basis $\{\psi_\alpha\}$ for $\mathbb{W}$ such that each element belongs to $\mathbb{W}$ can be written as the linear combination  of finite $\psi_\alpha$s. Let $E_\alpha$ be the support of $\psi_\alpha$, $E=\bigcup_\alpha E_\alpha$ and $\mathcal{D}(E)$ be the space of functions  vanishing on $\mathbb{Q}_p^n\backslash E$ in $\mathcal{D}(\mathbb{Q}_p^n)$. It follows  from the definition \eqref{1} that  $f_\alpha(\xi)=T_{(0,y)}[\psi](\xi)=\chi(\sum_iy_i\xi_i)\psi_\alpha(\xi)\in \mathbb{W}$ for any $y=(y_1,\cdots,y_n)\in\mathbb{Q}_p^n$. Obviously, $\mathbb{W}\subset \mathcal{D}(E)$, moreover we claim that $\mathbb{W}=\mathcal{D}(E)$. Otherwise there is a non-zero function $g\in\mathcal{D}(E) $ such that $(f_\alpha,g)=\int_{\mathbb{Q}_p^n}f_\alpha(\xi)\overline{g(\xi)}d^n\xi=F[\psi_\alpha\overline{g}](y)=0$ for $\forall\alpha$, which implies that $\psi_\alpha(\xi)\overline{g(\xi)}$ vanishes  on $\mathbb{Q}_p^n$ . Hence $g$ has to be zero  on $E$, but $g\equiv0$ on $\mathbb{Q}_p^n\backslash E$, so $g=0$, which exhibits  a contradiction. So far, we only need to show $E=\mathbb{Q}_p^n$ up to a zero measure set. Indeed, if  $\psi(\xi)\in \mathcal{D}(E)$, then $\psi(\xi+x)\in \mathcal{D}(E) $ for any $x=(x_1,\cdots,x_n)\in\mathbb{Q}_p^n$, namely $\psi(\xi+x)\in \mathcal{D}(E) $ vanishes outside $E$. Therefore $E$ is quasi-invariant with respect to the translation. It is known that $\mathbb{Q}_p^n$ as the Harr measure space is ergodic with respect to  the  translation group, then $\mathbb{Q}_p^n\backslash E$ is a zero measure set since $E$ cannot be a zero measure set.

Therefore the  Schr\"{o}dinger representation  is an irreducible smooth representation, hence to show it is admissible we only need to show it is supercuspidal\cite{p}. Indeed, since Schr\"{o}dinger representation $(\Phi,\mathcal{D}(\mathbb{Q}_p^n))$ and the induced representation $(\Phi^*,((\mathcal{D}(\mathbb{Q}_p^n))^*)$ are both  smooth, if one  fixes $\psi\in \mathcal{D}(\mathbb{Q}_p^n), \Psi\in (\mathcal{D}(\mathbb{Q}_p^n))^*$, then  there exist  a compact open subgroup $K$ of $H(V)$ such that $\psi\in (\mathcal{D}(\mathbb{Q}_p^n))^K, \Psi\in ((\mathcal{D}(\mathbb{Q}_p^n))^*)^K$ and $\Psi(\Phi(k_1gk_2)\psi)=\Psi(\Phi(g)\psi)$ for any $k_1, k_2\in K$, and $g\in H(V)$, where $(\mathcal{D}(\mathbb{Q}_p^n))^K=\{\psi\in\mathcal{D}(\mathbb{Q}_p^n):\Phi(g)\psi=\psi \textrm{ for any } g\in K\}$. Take  $g=((x,y);1)$ and $k_1( \textrm{ or }k_2 )=((x^\prime,y^\prime);1)$ with non-zero $x^\prime, y^\prime$, then $\Psi(\Phi(((x,y);1))\psi)=\chi(\frac{1}{2}(x^\prime \cdot y-x\cdot y^\prime))\Psi(\Phi(((x+x^\prime,y+y^\prime);1))\psi)=\chi(x^\prime \cdot y+\frac{1}{2}x^\prime\cdot y^\prime)\Psi(\Phi(((x,y);1))\psi)=\chi(y^\prime \cdot x+\frac{1}{2}x^\prime\cdot y^\prime)\Psi(\Phi(((x,y);1))\psi)$ which shows the supercuspidality.
\end{proof}

\begin{theorem}(Non-Archimedean Stone-von Neumann Theorem\cite{mv}) Any smooth representation $\Phi$ of  the Heisenberg group $H(V)$ satisfying $\Phi(0,\alpha)=\alpha Id$ decomposes into the direct sum of irreducible representations equivalent to the  Schr\"{o}dinger representations.
\end{theorem}

Let $A:\mathcal{D}(\mathbb{Q}_p^n)\rightarrow\mathcal{D}^*(\mathbb{Q}_p^n)$ be a linear operator, then we define a bilinear functional $L_A$ on $\mathcal{D}(\mathbb{Q}_p^n)$ associated to $A$ as follows:
\begin{align*}
  L_A(\psi,\varphi)=\int\alpha^{\frac{1}{2}} \chi_{\frac{1}{8}}(\sum_i(x_i^2-y_i^2))\langle\Phi^*(g)\circ A(\psi),\varphi\rangle d\mu_g
\end{align*}
for $\psi,\varphi\in \mathcal{D}(\mathbb{Q}_p^n)$, where the linear operator $\Phi(g)^*:\mathcal{D}^*(\mathbb{Q}_p^n)\rightarrow\mathcal{D}^*(\mathbb{Q}_p^n)$ is induced by $\Phi(g)$, $\langle,\rangle$ stands for the pairing between $\mathcal{D}^*(\mathbb{Q}_p^n)$ and $\mathcal{D}(\mathbb{Q}_p^n)$, and
$d\mu_g$ denotes the invariant measure on $H(V)$.
Hence there exists a a linear operator $P_A:\mathcal{D}(\mathbb{Q}_p^n)\rightarrow\mathcal{D}^*(\mathbb{Q}_p^n)$ such that $L_A(\psi,\varphi)=\langle P_A\psi,\varphi\rangle$\cite{ak}.
\begin{proposition}$L_{P_{A}}=CL_{A}$, where $C$ is a constant.
\end{proposition}
\begin{proof}
According to the definition, we calculate
\begin{align*}
& L_{\Phi(g)^*\circ P_{A}}(\psi,\varphi)=\int (\alpha^\prime)^{\frac{1}{2}}\chi_{\frac{1}{8}}(\sum_i((x^\prime_i)^2-(y^\prime_i)^2))\langle\Phi(g^\prime)^*\circ\Phi(g)^*\circ P_A(\psi),\varphi\rangle d\mu_{g^\prime}\\
=&\int (\alpha^\prime)^{\frac{1}{2}}\chi_{\frac{1}{8}}(\sum_i((x^\prime_i)^2-(y^\prime_i)^2))\langle P_A(\psi),\Phi(g)\circ\Phi(g^\prime)(\varphi)\rangle d\mu_{g^\prime}\\
  =&\int(\alpha^\prime)^{\frac{1}{2}}(\alpha^{\prime\prime})^{\frac{1}{2}}\chi_{\frac{1}{8}}(\sum_i((x^\prime_i)^2-(y^\prime_i)^2+(x_i^{\prime\prime})^2-(y^{\prime\prime}_i)^2)) \langle\Phi(g^\prime)^*\circ
\Phi(g)^*\circ\Phi( g^{\prime\prime})^*\circ A(\psi),\varphi\rangle d\mu_{g^\prime } d\mu_{g^{\prime\prime}}.
  \end{align*}
  Let $\tilde g_0=gg^{\prime}, \tilde g=g^{\prime\prime} \tilde g_0$; and  $X_i=x_i+x^{\prime}_i$, $\tilde X_i=x^{\prime\prime}_i+X_i$, $Y_i=y_i+y^{\prime}_i$, $\tilde Y_i=y^{\prime\prime}_i+Y_i$, then we arrive at
  \begin{align*}
  & L_{\Phi(g)^*\circ P_{A}}(\psi,\varphi)\\
  =&\int(\alpha^{\prime\prime})^{\frac{1}{2}}(\alpha^{\prime})^{\frac{1}{2}}\chi_{\frac{1}{8}}(\sum_i((\tilde X_i-X_i)^2-(\tilde Y_i-Y_i)^2+(X_i-x_i)^2-(Y_i-y_i)^2))\langle\Phi(\tilde g)^*\circ A(\psi),\varphi\rangle d\mu_{\tilde g }d\mu_{ \tilde g_0}\\
  =& \alpha^{-\frac{1}{2}}\chi_{\frac{1}{8}}(\sum_i(x_i^2-y_i^2))\int \int\chi_{\frac{1}{4}}(\sum_i(X_i^2-(\tilde X_i+x_i-\tilde Y_i-y_i)X_i-Y_i^2+(\tilde Y_i+y_i-\tilde X_i-x_i)Y_i))d\mu_{ \tilde g_0}\\
  &\cdot(\alpha\alpha^{\prime\prime}\alpha^{\prime}\chi(\frac{1}{2}\sum_i(x_iy_i^{\prime}+x_i^{\prime\prime} y_i+x_i^{\prime\prime} y_i^{\prime}-y_ix_i^{\prime}-y_i^{\prime\prime} x_i-y_i^{\prime\prime} x_i^{\prime})))^{\frac{1}{2}}\chi_{\frac{1}{8}}(\sum_i(\tilde X_i^2-\tilde Y_i^2))\langle\Phi(\tilde g)^*\circ A(\psi),\varphi\rangle d\mu_{\tilde g }\\
  =&C\Theta(g)L_A(\psi,\varphi),
\end{align*}
where \begin{align*}
       \Theta(g)&=\alpha^{-\frac{1}{2}}\chi_{\frac{1}{8}}(\sum_i(x_i^2-y_i^2)), \\ C&=\textrm{Vol}(\mathbf{S}^1)\frac{\lambda_p(\frac{1}{4})\lambda_p(-\frac{1}{4})}{\sqrt{|-\frac{1}{16}|_p}}=\textrm{Vol}(\mathbf{S}^1)|4|_p\lambda_p(\frac{1}{4})\lambda_p(-\frac{1}{4}).
      \end{align*}
 Consequently, we find that $L_{\Phi(g)^*\circ P_{A}}=C\Theta(g)L_A$, which implies that  $L_{ P_{A}}=CL_A$  when one takes $g=Id$.
 \end{proof}

We can define the so-called symplectic Fourier transformation for $\psi\in\mathcal{D}(\mathbb{Q}_p^{2n})$ as follows: $$F_s[\psi](z)=\check{\psi}(z)=\int\chi(J_0(z,z^\prime))\psi(z^\prime)d^{2n}z^\prime,$$
which is related with the usual Fourier transformation  via the formula $F_s[\psi](J_0z)=F[\psi](z)$. Therefore $F_s$ as an operator on $\mathcal{D}(\mathbb{Q}_p^{2n})$ extends into a unitary operator on $L^2(\mathbb{Q}_p^{2n})$. The Weyl operator $W_f$ associated to the symbol $f\in \mathcal{D}(\mathbb{Q}_p^{2n})$ is defined by
\begin{align*}
 W_f[\psi](\xi)&=\int \check{f}(z)T_{(z,1)}[\psi](\xi)d^{2n}z\\
&=\int \check{f}(z)\psi(\xi+x)\chi(\sum_i(y_i\xi_i+\frac{1}{2}x_iy_i))d^{2n}z.
\end{align*}
\begin{proposition} $W_f$ is a linear and continuous  operator on $\mathcal{D}(\mathbb{Q}_p^{n})$, namely $W_f[\psi]\in\mathcal{D}(\mathbb{Q}_p^{n})$ for $\psi\in\mathcal{D}(\mathbb{Q}_p^{n})$.
 \end{proposition}
\begin{proof}
From $\check{f}\in \mathcal{D}(\mathbb{Q}_p^{2n})$, it follows that there exists $z^\prime=(0,\cdots,0, a_1, \cdots, a_n )\in\mathbb{Q}_p^{2n}$ with $|a_i|_p=p^l(i=1,\cdots,n)$ for an integer $l$ such that $\check{f}(z)=\check{f}(z+z^\prime)$. Then since
$W_f[\psi](\xi)=\chi(\sum_i\xi_ia_i)W_f[\psi](\xi)$, and since if $|\xi|_p>p^{-l}$, there exists $k\in\{1,\cdots,n\}$ such that $|\xi_ka_k|_p>1$, thus $\chi(\sum_i\xi_ia_i)\neq 1$, we observe that $\textrm{supp}(W_f[\psi])\subset (B(0;l))^n$. To show $W_f[\psi]$ is a locally constant function, we consider
$$W^N_f[\psi](\xi):=\int_{(B(0;-N))^{2n}} \check{f}(z)\psi(\xi+x)\chi(\sum_i(y_i\xi_i+\frac{1}{2}x_iy_i))d^{2n}z$$
for an integer $N$.  Let $l\in \mathbb{Z}$ be the largest characteristic number associated with $\psi$. Then for any $\xi^\prime\in (B(0;N))^n$ we have
$W^N_f[\psi](\xi+\xi^\prime)=W^N_f[\psi](\xi)$ if $N>-l$  because  $|\sum_iy_i\xi^\prime_i|_p\leq \textrm{max}\{|y_i\xi^\prime_i|_p\}\leq1$ and $\psi(\xi+x+\xi^\prime)=\psi(\xi+x)$ when $N>-l$. Consequently, $W_f[\psi]=\lim_{N\rightarrow\infty}W^N_f[\psi]$ is locally constant.
\end{proof}

\begin{proposition}
\begin{enumerate}
                     \item $W_f[\psi](\xi)=\int K_{W_f}(\xi,\eta)\psi(\eta)d^n\eta$, where
 $$K_{W_f}(\xi,\eta)=\int f(\frac{1}{2}(\eta+\xi),y)\chi(\sum_i(\eta_i-\xi_i)y_i)d^{n}y$$
   is called the kernel function   of Weyl operator.
\item $||W_f[\psi]||_{L^2}\leq ||F[f]||_{L^1}||\psi||_{L^2}$.
\item For $f,g\in \mathcal{D}(\mathbb{Q}_p^n)$, we have the composition law
$K_{W_f\circ W_g}=K_{W_h}$,
where  $$h(z)=c\int f(z+z^\prime)g(z+z^{\prime\prime})\chi(-2J_0(z^\prime,z^{\prime\prime}))d^{2n}z^{\prime}d^{2n}z^{\prime\prime}\in \mathcal{D}(\mathbb{Q}_p^n)$$
with $c=\left\{
         \begin{array}{ll}
           1, & \hbox{$p\geq3$;} \\
           \frac{1}{4^n}, & \hbox{$p=2$.}
         \end{array}
       \right.
$
                   \end{enumerate}

\end{proposition}
\begin{proof}(1) Let $\eta=\xi+x$, then we have
\begin{align*}
  W_f[\psi](\xi)&=\int\check{f}(\eta-\xi,y)\psi(\eta)\chi(\sum_i\frac{1}{2}y_i(\eta_i+\xi_i))d^n\eta d^ny\\
&=\int f(z^\prime)\psi(\eta)\chi(\sum_i(\eta_i-\xi_i)y^\prime_i)\chi(\sum_i\frac{1}{2}y_i(\eta_i+\xi_i-2x^\prime_i))d^{2n}z^\prime d^n\eta d^ny\\
&=\int f(z^\prime)\psi(\eta)\chi(\sum_i(\eta_i-\xi_i)y^\prime_i)\delta(\frac{1}{2}(\eta_i+\xi_i-2x^\prime_i))d^{2n}z^\prime d^n\eta\\
&=\int f(\frac{1}{2}(\eta+\xi),y^\prime)\psi(\eta)\chi(\sum_i(\eta_i-\xi_i)y^\prime_i)d^{n}y^\prime d^n\eta
\end{align*}
where we have applied the Fourier transform of Dirac distribution $\delta$ for the third equality\cite{ak}: $F[\delta]=1$ and $F[1]=\delta$.

(2) First we note that $F[f]\in L^1(\mathbb{Q}_p^{2n})$ since $F[f]$ belongs to $\mathcal{D}(\mathbb{Q}_p^{2n})$ which is dense in $L^1(\mathbb{Q}_p^{2n})$\cite{ak}. By  manipulation of the Fourier transform of Dirac distribution once again, we express the kernel function as
 \begin{align*}
    K_{W_f}(\xi,\eta)&=\int f(k,y)  \chi(\sum_i(\eta_i-\xi_i)y_i)\chi(\sum_i\frac{1}{2}x_i(2k_i-\eta_i-\xi_i))d^nk d^nyd^nx\\
&=  \int F[f](x,\eta-\xi)  \chi (-\sum_i\frac{1}{2}x_i(\eta_i+\xi_i))  d^nx,                                                                             \end{align*}
which leads to the desired inequalities
\begin{align*}
 \int| W_f[\psi](\xi)|^2d^n\xi&\leq\int(\int |K_{W_f}(\xi,\eta)|d^n\eta)(\int |K_{W_f}(\xi,\eta)|\cdot |\psi(\eta)|^2d^n\eta)d^n\xi\\
&\leq(\int |F[f](\xi,\eta)|d^n\xi d^n\eta)^2\int  |\psi(\eta)|^2 d^n\eta.
\end{align*}

(3) One can easily check that $ K_{W_f\circ W_g}(\xi,\eta)=\int K_{W_f}(\xi,\zeta)K_{W_g}(\zeta,\eta)d^n\zeta$. Then since
\begin{align*}
 & \int K_{W_f}(\xi,\zeta)K_{W_g}(\zeta,\eta)d^n\zeta\\=&\int f(\frac{1}{2}(\xi+\zeta),x)g(\frac{1}{2}(\zeta+\eta),y)\chi(\sum_i(\zeta_i-\xi_i)x_i)\chi(\sum_i(\eta_i-\zeta_i)y_i)d^nxd^nyd^n\zeta,
\end{align*}
we get
\begin{align*}
  h(\xi,\eta)&=\int K_{W_f\circ W_g}(\xi-\frac{1}{2}x,\xi+\frac{1}{2}x)\chi(-\sum_i\eta_ix_i)d^nx\\
&=\int f(\frac{1}{2}(\xi-\frac{1}{2}x+\zeta),\alpha)g(\frac{1}{2}(\xi+\frac{1}{2}x+\zeta),\beta)\\
&\ \ \ \ \ \ \ \chi(\sum_i(\zeta_i-\xi_i+\frac{1}{2}x)\alpha_i)
\chi(\sum_i(\xi_i-\zeta_i+\frac{1}{2}x)\alpha_i)\chi(-\sum_i\eta_ix_i)d^n\alpha d^n\beta d^n\zeta d^nx\\
&=|-\frac{1}{4}|^{-n}_p\int f(\xi+u,\eta+\alpha^\prime)g(\xi+v,\eta+\beta^\prime)\chi(2\sum_i(v_i\alpha^\prime_i-u_i\beta^\prime_i))d^n\alpha^\prime d^n\beta^\prime d^nu d^nv,
\end{align*}
where the constant $|-\frac{1}{4}|^{-n}_p=|4|^{n}_p$ in front of the integral comes from the variable changing $u=\frac{1}{2}(-\frac{1}{2}x-\xi+\eta)$, $v=\frac{1}{2}(\frac{1}{2}x-\xi+\eta)$, $\alpha^\prime=\alpha-\eta$ and $\beta^\prime=\beta-\eta$.
\end{proof}

\subsection{Induced representations and Maslov indices}
Let $\Gamma$ be an abelian  subgroup  of the Heisenberg group $H(V)$, then we have a smooth representation $\Upsilon(\Gamma,\mathcal{C})$ of $H(V)$ induced by the given unitary character $\mathcal{C}$ of $\Gamma$   restricting to the
identity on  the center $\mathbf{S}^1$ of $H(V)$ and being locally constant  restricted on $V$, namely $\Upsilon(L,\mathcal{C})=\textrm{Ind}_\Gamma^{H(V)}\mathcal{C}$,
which can be realized as follows. Let $\mathbb{H}(\Gamma,\mathcal{C})$
be the set of complex valued  functions on $H(V)$ which are locally constant compactly supported when restricted  on $V$ and satisfy the conditions $\psi(\gamma g)=\mathcal{C}(\gamma)\psi(g)$ for any $g\in H(V),\gamma\in \Gamma$. Then the induced  representation $\Upsilon(\Gamma,\mathcal{C})$  is defined to be the representation
of $H(V)$ in $\mathbb{H}(\Gamma,\mathcal{C})$ given by right translations:
$\Upsilon(\Gamma,\mathcal{C})(g_0)[\psi](g)=\psi(g g_0)$ for $\psi\in \mathbb{H}$, $g_0, g\in H(V)$. Obviously, $\Upsilon(\Gamma,\mathcal{C})(\alpha)=\alpha Id$ for $\alpha\in \mathbf{S}^1$.  One can show that $(\Upsilon(\Gamma,\mathcal{C}),\mathbb{H})$ is an irreducible\cite{ho} and admissible representation.
 Stone-von Neumann theorem implies that every induced representation $\Upsilon(\Gamma,\mathcal{C})$ is isomorphic to the Schr\"{o}dinger representation, thus there exists a  unique unitary isomorphism $\Theta(\Gamma,\mathcal{C};\Gamma^\prime,\mathcal{C}^\prime):\mathbb{H}(\Gamma,\mathcal{C})\rightarrow\mathbb{H}(\Gamma^\prime,\mathcal{C}^\prime)$ determined by the relation
$\Theta(\Gamma,\mathcal{C};\Gamma^\prime,\mathcal{C}^\prime)\Upsilon(\Gamma,\mathcal{C})\Theta^{-1}(\Gamma,\mathcal{C};\Gamma^\prime,\mathcal{C}^\prime)=\Upsilon(\Gamma^\prime,\mathcal{C}^\prime)$ up to a unit scalar  related to $\Gamma,\Gamma^\prime$.
This isomorphism is called the intertwining operator between the two corresponding representations.

For example, we can take $\Gamma=(L,\mathbf{S}^1)=:\Gamma(L)$ or $\Gamma=(\ell,\mathbf{S}^1)=:\Gamma(\ell)$ for any self-dual lattice $L$ or Lagrangian subspace $\ell$ in $(V,J)$. By the restriction from  $H(V)$ to $V$, the function $\psi\in \mathbb{H}(\Gamma,\mathcal{C})$ descends to a function on $V$ denoted still by $\psi$ that satisfies $\psi(z+\gamma)=\chi(\frac{1}{2}J(z,\gamma))\mathcal{C}((\gamma;1))\psi(z)$ for $z\in V,\gamma\in L \textrm{ or } \ell$,
and $\Upsilon(\Gamma,\mathcal{C})(g_0)[\psi](z)=\alpha\chi(\frac{1}{2}J(z,w))\psi(z+w)$ for $g_0=(w;\alpha)\in H(V)$. For a self-dual lattice $L$, there exist two transversal Lagrangian subspaces $\ell$ and $\ell^\prime$, where the transversality condition means  that $\ell\cap\ell^\prime=\{0\}$,  such that $L=\ell\cap L\oplus \ell^\prime\cap L$. The characters $\mathcal{C}_\ell, \mathcal{C}_L$ of $\Gamma(\ell)$ and $\Gamma(L)$ are specified such that
$\mathcal{C}_L((\ell\cap L;1))=\mathcal{C}_\ell((\ell\cap L;1))$.
\begin{proposition}The isomorphism $\Theta(\Gamma(\ell),\mathcal{C}_\ell;\Gamma(L),\mathcal{C}_L):\mathbb{H}(\Gamma(\ell),\mathcal{C}_\ell)\rightarrow\mathbb{H}(\Gamma(L),\mathcal{C}_{L})$ is given by
\begin{align*}
 \Theta(\Gamma(\ell),\mathcal{C}_\ell;\Gamma(L),\mathcal{C}_L)[\psi](g)\sim\sum_{u\in\ell^\prime \cap L}\mathcal{C}_L((-u;1))\psi((u;1)\cdot g).
\end{align*}
where $\sim$ means the equality up to a constant of modulus one related to  $L$ and $\ell$.
\end{proposition}

\begin{proof}We first need to show $\Theta(\Gamma(\ell),\mathcal{C}_\ell;\Gamma(L),\mathcal{C}_L)[\psi]\in \mathbb{H}(\Gamma(L),\mathcal{C}_L)$ for $\psi\in \mathbb{H}(\Gamma(\ell),\mathcal{C}_\ell)$. If one chooses a suitable symplectic basis $\{e_1,\cdots,e_n,f_1,\cdots,f_n\}$ of $V$ such that $L=\mathbb{Z}_pe_1\oplus\cdots\oplus\mathbb{Z}_pe_n\oplus\mathbb{Z}_pf_1\oplus\cdots\oplus\mathbb{Z}_pf_n$, $\ell=\mathbb{Q}_pe_1\oplus\cdots\oplus\mathbb{Q}_pe_n$ and
$\ell^\prime=\mathbb{Q}_pf_1\oplus\cdots\oplus\mathbb{Q}_pf_n$, then for $g=((x,y);\alpha)$ we express
\begin{align*}
  \Theta(\Gamma(\ell),\mathcal{C}_\ell;\Gamma(L),\mathcal{C}_L)[\psi](g)
  \sim\alpha\mathcal{C}_\ell(x)\chi(-\sum_i\frac{1}{2}x_iy_i)\sum_{u\in\mathbb{Z}_p^n}\mathcal{C}_L(-u)
  \chi(-\sum_iu_ix_i)\psi(u+y),
\end{align*}
where  $\mathcal{C}_\ell(x):=\mathcal{C}_\ell(((x,0);1))$, $\mathcal{C}_L(u):=\mathcal{C}_L(((0,u);1))$ and  $\psi(u+y):=\psi(((0,u+y);1))$.
The condition  that $\Theta(\Gamma(\ell),\mathcal{C}_\ell;\Gamma(L),\mathcal{C}_L)[\psi](\gamma g)=\mathcal{C}_L(\gamma)\Theta(\Gamma(\ell),\mathcal{C}_\ell;\Gamma(L),\mathcal{C}_L)[\psi](g)$ for any $\gamma\in\Gamma(L)$ can be easily  verified.
Obviously $\Theta(\Gamma(\ell),\mathcal{C}_\ell;\Gamma(L),\mathcal{C}_L)[\psi]$ is locally constant restricted on $V$ since the characteristic functions associated with locally constant functions  $\chi(-\sum_iu_ix_i)$ and $\psi(u+y)$ on $\mathbb{Q}_p^n$ can be taken independent of $u$. On the other hand, there exists $a\in \mathbb{Z}_p^n$ with $|a|_p=p^l$ for some integer $l$ such that $\Theta(\Gamma(\ell),\mathcal{C}_\ell;\Gamma(L),\mathcal{C}_L)[\psi](g)=\mathcal{C}_L((0,a);1)\chi(\sum_ix_ia_i)\Theta(\Gamma(\ell),\mathcal{C}_\ell;\Gamma(L),\mathcal{C}_L)[\psi](g)$, meanwhile, we have  $\{y\in \mathbb{Q}_p^n:\Theta(\Gamma(\ell),\mathcal{C}_\ell;\Gamma(L),\mathcal{C}_L)[\psi]((x,y);\alpha)\neq0\}\subset\bigcup_{u\in\mathbb{Z}_p^n}\{y-u:\psi(((0,y);1))\neq0\}$, the latter set is bounded in $\mathbb{Q}_p^n$. Therefor $\Theta(\Gamma(\ell),\mathcal{C}_\ell;\Gamma(L),\mathcal{C}_L)[\psi]$ is compactly supported on $V$. Next we should prove $\Theta(\Gamma(\ell),\mathcal{C}_\ell;\Gamma(L),\mathcal{C}_L)$  is surjective. To show it, we define the intertwining operator $\Theta(\Gamma(L),\mathcal{C}_L;\Gamma(\ell),\mathcal{C}_\ell):\mathbb{H}(\Gamma(L),\mathcal{C}_L)\rightarrow\mathbb{H}(\Gamma(\ell),\mathcal{C}_{\ell})$ by
\begin{align*}
 \Theta(\Gamma(L),\mathcal{C}_L;\Gamma(\ell),\mathcal{C}_\ell)[\psi](g)\sim&\int_{\ell/\ell\cap L}\mathcal{C}_\ell((-u;1))\psi((u;1)\cdot g)d\mu_{\ell/\ell\cap L}\\
 &=\int_{\mathbb{Q}_p^n/\mathbb{Z}_p^n}\alpha\mathcal{C}_\ell(-u)\chi(\frac{1}{2}\sum_iu_iy_i)\psi(u+x,y)d^nu\\
 &=\alpha\mathcal{C}_\ell(x)\chi(-\sum_i\frac{1}{2}x_iy_i)\int_{\mathbb{Q}_p^n/\mathbb{Z}_p^n}\mathcal{C}_\ell(-u)\chi(\frac{1}{2}\sum_iu_iy_i)\psi(u,y)d^nu,
\end{align*}
where regarding $u$ as a representative of an equivalent class in $\ell/\ell\cap L$ makes sense, and   $\psi(u+x,y):=\psi(((u+x,y);1))$. The similar arguments implies that $\Theta(\Gamma(L),\mathcal{C}_L;\Gamma(\ell),\mathcal{C}_\ell)[\psi]\in\mathbb{H}(\Gamma(\ell),\mathcal{C}_\ell)$.
We also  have to show $\Theta(\Gamma(\ell),\mathcal{C}_\ell;\Gamma(L),\mathcal{C}_L)\circ\Theta(\Gamma(L),\mathcal{C}_L;\Gamma(\ell),\mathcal{C}_\ell)\sim Id$. Indeed, for $\psi\in \mathbb{H}(\Gamma(L),\mathcal{C}_L)$ we calculate
\begin{align*}
  &\Theta(\Gamma(\ell),\mathcal{C}_\ell;\Gamma(L),\mathcal{C}_L)\circ\Theta(\Gamma(L),\mathcal{C}_L;\Gamma(\ell),\mathcal{C}_\ell)[\psi](g)\\ \sim&
  \alpha\sum_{v\in\mathbb{Z}_p^n}\int_{\mathbb{Q}_p^n/\mathbb{Z}_p^n}\mathcal{C}_\ell(-u)\chi(\sum_i(u_iv_i+\frac{1}{2}u_iy_i))\psi(u+x,y)d^nu\\
  &=\psi((x,y);\alpha).
\end{align*}
Finally, the isometry property of these isomorphisms can been immediately seen from the Parseval-Steklov equality for Fourier transformations.
\end{proof}
\begin{corollary}Let $\ell_1, \ell_2$ be two  Lagrangian subspaces in $(V,J)$, and the characters $\mathcal{C}_{\ell_1}$, $\mathcal{C}_{\ell_1}$ of $\Gamma(\ell_1)$, $\Gamma(\ell_2)$  be chosen to satisfies  $\mathcal{C}_{\ell_1}((\ell_1\cap\ell_2;1))=\mathcal{C}_{\ell_2}((\ell_1\cap\ell_2;1))$, then the isomorphism $\Theta(\Gamma(\ell_1),\mathcal{C}_{\ell_1};\Gamma(\ell_2),\mathcal{C}_{\ell_2}):\mathbb{H}(\Gamma(\ell_1),\mathcal{C}_{\ell_1})\rightarrow\mathbb{H}(\Gamma(\ell_2),\mathcal{C}_{\ell_2})$ is given by
\begin{align}\label{eq:p}
\Theta(\Gamma(\ell_1),\mathcal{C}_{\ell_1};\Gamma(\ell_2),\mathcal{C}_{\ell_2})[\psi](g)\sim&\int_{\ell_2/\ell_1\cap\ell_2}\mathcal{C}_{\ell_2}((-u;1))\psi((u;1)\cdot g)d\mu_{\ell_2/\ell_1\cap\ell_2}\\
&=\alpha\mathcal{C}_{\ell_2}(x)\chi(-\sum_{i=1}^n\frac{1}{2}x_iy_i)\int_{\mathbb{Q}_p^m}\mathcal{C}_{\ell_2}(-u)\chi(-\sum_{i=1}^mu_ix_i)\psi(u+y)d^mu,\nonumber
\end{align}
where the measure $\mu_{\ell_2/\ell_1\cap\ell_2}$ on $\ell_2/\ell_1\cap\ell_2$ has to be suitably chosen to guarantee the unitarity of intertwining operators, and the second equality is expressed in terms of the suitable symplectic basis of $(V,J)$ with $m=n-\dim_{\mathbb{Q}_p}\ell_1\cap\ell_2$ and $\mathcal{C}_{\ell_2}(u)=\mathcal{C}_{\ell_2}(((u_1,\cdots,u_m,0,\cdots,0);1))$, $\psi(u+y)=\psi(((u_1,\cdots,u_m,0,\cdots,0,y_1,\cdots,y_n);1))$.
                   \end{corollary}
\begin{proof}There exists a
   symplectic basis $\{e_1,\cdots,e_n,f_1,\cdots,f_n\}$ such that $\ell_1=\mathbb{Q}_pe_1\oplus\cdots\oplus\mathbb{Q}_pe_n$ and $\ell_2=\mathbb{Q}_pf_1\oplus\mathbb{Q}_pf_{m}\oplus\mathbb{Q}_pe_{m+1}\cdots\oplus\mathbb{Q}_pe_n$.  We fix a self-dual lattice $L=\mathbb{Z}_pe_1\oplus\cdots\oplus\mathbb{Z}_pe_n\oplus\mathbb{Z}_pf_1\oplus\cdots\oplus\mathbb{Z}_pf_n$. Then $\ell_1\cap\ell_2=\mathbb{Q}_pe_{m+1}\cdots\oplus\mathbb{Q}_pe_n$ and $L_0=\mathbb{Z}_pe_{m+1}\cdots\oplus\mathbb{Z}_pe_n\oplus\mathbb{Z}_pf_{m+1}\cdots\oplus\mathbb{Z}_pf_n$ are Lagrangian subspace and self-dual lattice in $(V_0=\mathbb{Q}_pe_{m+1}\cdots\oplus\mathbb{Q}_pe_n\oplus\mathbb{Q}_pf_{m+1}\cdots\oplus\mathbb{Q}_pf_n,J|_{V_0})$ respectively. The isomorphism
   $\Theta(\Gamma(\ell_1),\mathcal{C}_{\ell_1};\Gamma(\ell_2),\mathcal{C}_{\ell_2})$ can be constructed via the composition
\begin{align*}
  \Theta(\Gamma(\ell_1),\mathcal{C}_{\ell_1};\Gamma(\ell_2),\mathcal{C}_{\ell_2})\sim
  \Theta(\Gamma(L),\mathcal{C}_{L};\Gamma(\ell_2),\mathcal{C}_{\ell_2})\circ\Theta(\Gamma(\ell_1),\mathcal{C}_{\ell_1};\Gamma(L),\mathcal{C}_{L}),
\end{align*}
where the character $\mathcal{C}_L$ is fixed with the properties that $\mathcal{C}_L((\ell_{1,2}\cap L;1))=\mathcal{C}_{\ell_{1,2}}((\ell_{1,2}\cap L;1))$. On the other hand, we observe that
\begin{align*}
 &\Theta(\Gamma(\ell_1),\mathcal{C}_{\ell_1};\Gamma(L),\mathcal{C}_{L})\\
 \sim&\Theta(\Gamma(\ell_1\cap\ell_2),\mathcal{C}_{\ell_1}|_{\ell_1\cap\ell_2};\Gamma(L_0),\mathcal{C}_{L}|_{L_0})\otimes
\Theta(\Gamma(\ell_1^\prime),\mathcal{C}_{\ell_1}|_{\ell_1^\prime};\Gamma(L^\prime),\mathcal{C}_{L}|_{L^\prime}),\\
&\Theta(\Gamma(L),\mathcal{C}_{L};\Gamma(\ell_2),\mathcal{C}_{\ell_2})\\
\sim&\Theta(\Gamma(L_0),\mathcal{C}_{L}|_{L_0};\Gamma(\ell_1\cap\ell_2),\mathcal{C}_{\ell_2}|_{\ell_1\cap\ell_2})\otimes
\Theta(\Gamma(L^\prime),\mathcal{C}_{L}|_{L^\prime};\Gamma(\ell_2^\prime),\mathcal{C}_{\ell_2}|_{\ell_2^\prime}).
\end{align*}
where $\ell^\prime_{1}=\ell_{1}/\ell_1\cap\ell_2=\mathbb{Q}_pe_1\oplus\cdots\oplus\mathbb{Q}_pe_m$, $\ell^\prime_{2}=\ell_{2}/\ell_1\cap\ell_2=\mathbb{Q}_pf_1\oplus\cdots\oplus\mathbb{Q}_pf_m$ and $L^\prime=L/L_0=\mathbb{Z}_pe_1\oplus\cdots\oplus\mathbb{Z}_pe_m\oplus\mathbb{Z}_pf_1\oplus\cdots\oplus\mathbb{Z}_pf_m$.
As a consequence, we arrive at \begin{align*}
                               &\Theta(\Gamma(\ell_1),\mathcal{C}_{\ell_1};\Gamma(\ell_2),\mathcal{C}_{\ell_2})[\psi](g)\\
                               \sim&
                               \Theta(\Gamma(L^\prime),\mathcal{C}_{L}|_{L^\prime};\Gamma(\ell_2^\prime),\mathcal{C}_{\ell_2}|_{\ell_2^\prime})
                               \circ\Theta(\Gamma(\ell_1^\prime),\mathcal{C}_{\ell_1}|_{\ell_1^\prime};\Gamma(L^\prime),\mathcal{C}_{L}|_{L^\prime})[\psi](g)\\
                                \sim&\int_{\ell_2^\prime/\ell_2^\prime\cap L^\prime}\mathcal{C}_{\ell_2}|_{\ell_2^\prime}((-u;1))\sum_{v\in L^\prime/L^\prime\cap\ell_1^\prime }\mathcal{C}_{L}|_{L^\prime}((-v;1))\psi((v;1)\cdot(u;1)\cdot g)d\mu_{\ell_2^\prime/\ell_2^\prime\cap L^\prime}\\
                                &=\int_{\ell_2^\prime}\mathcal{C}_{\ell_2}|_{\ell_2^\prime}((-u;1))\psi((u;1)\cdot g)d\mu_{\ell_2^\prime}.
                              \end{align*}
                              Thus, we complete the proof.
                              \end{proof}
\begin{definition}The isomorphism as the r.h.s. of \eqref{eq:p} is called the canonical isomorphism between the representation spaces $ \mathbb{H}(\Gamma(\ell_1),\mathcal{C}_{\ell_1})$ and $\mathbb{H}(\Gamma(\ell_2),\mathcal{C}_{\ell_2})$, and denoted by $\Theta_c(\Gamma(\ell_1),\mathcal{C}_{\ell_1};\Gamma(\ell_2),\mathcal{C}_{\ell_2})$. For a triple $(\ell_1,\ell_2,\ell_3)$ of Lagrangian subspaces in $(V,J)$ and the characters $\mathcal{C}^0_{\ell_i}((\ell_i;1))=1 (i=1,2,3)$, we have
\begin{align}
  &\Theta_c(\Gamma(\ell_3),\mathcal{C}^0_{\ell_3};\Gamma(\ell_1),\mathcal{C}^0_{\ell_1})\circ\Theta_c(\Gamma(\ell_2),\mathcal{C}^0_{\ell_2};\Gamma(\ell_3),\mathcal{C}^0_{\ell_3})\circ\Theta_c(\Gamma(\ell_1),\mathcal{C}^0_{\ell_1};\Gamma(\ell_2),\mathcal{C}^0_{\ell_2})
\nonumber\\
=&\alpha(\ell_1,\ell_2,\ell_3)Id,
\end{align}
and the coefficient $\alpha(\ell_1,\ell_2,\ell_3)\in\mathbf{S}^1$ is called the Maslov index associated with $(\ell_1,\ell_2,\ell_3)$, or equivalently, the Maslov index is determined by
\begin{align}
  \Theta_c(\Gamma(\ell_2),\mathcal{C}^0_{\ell_2};\Gamma(\ell_3),\mathcal{C}^0_{\ell_3})\circ\Theta_c(\Gamma(\ell_1),\mathcal{C}^0_{\ell_1};\Gamma(\ell_2),\mathcal{C}^0_{\ell_2})
=\alpha(\ell_1,\ell_2,\ell_3)\Theta_c(\Gamma(\ell_1),\mathcal{C}^0_{\ell_1};\Gamma(\ell_3),\mathcal{C}^0_{\ell_3}).
\end{align}
\end{definition}

\begin{proposition}
Let $\ell_1,\ell_2, \ell_3, \ell_4$ be four Lagrangian subspaces in $(V,J)$, then the Maslov index is given by
\begin{align*}
  \alpha(\ell_1,\ell_2,\ell_3)=\int_{v\in\ell_3/\ell_2\cap\ell_3,w\in\ell_2/\ell_1\cap\ell_2,v+w\in\ell_1}\chi(\frac{1}{2}J(w,v))d\mu_{\ell_3/\ell_2\cap\ell_3}d\mu_{\ell_2/\ell_2\cap\ell_1},
\end{align*}
and it has the following properties:\begin{itemize}
                                      \item (permutation relations) $\alpha(\ell_1,\ell_2,\ell_3)=\overline{\alpha(\ell_1,\ell_3,\ell_2)}$, $\alpha(\ell_1,\ell_2,\ell_3)=\alpha(\ell_2,\ell_3,\ell_1)$,
                                      \item (cocycle relations) $\alpha(\ell_1,\ell_2,\ell_3)\alpha(\ell_1,\ell_3,\ell_4)\alpha(\ell_2,\ell_4,\ell_3)\alpha(\ell_2,\ell_1,\ell_4)=1$.
                                    \end{itemize}
\end{proposition}

\begin{proof}According to the definitions, for any $\psi\in \mathbb{H}(\Gamma(\ell_1),\mathcal{C}^0_{\ell_1})$ and $g=(z;\alpha)\in H(V)$, we have
\begin{align*}
   &\Theta_c(\Gamma(\ell_3),\mathcal{C}^0_{\ell_3};\Gamma(\ell_1),\mathcal{C}^0_{\ell_1})\circ\Theta_c(\Gamma(\ell_2),\mathcal{C}^0_{\ell_2};\Gamma(\ell_3),\mathcal{C}^0_{\ell_3})\circ\Theta_c(\Gamma(\ell_1),\mathcal{C}^0_{\ell_1};\Gamma(\ell_2),\mathcal{C}^0_{\ell_2})[\psi](g)\\
   =&\int_{\ell_1/\ell_1\cap\ell_3}d\mu_{\ell_1/\ell_1\cap\ell_3}\int_{\ell_3/\ell_2\cap\ell_3}d\mu_{\ell_3/\ell_2\cap\ell_3}\int_{\ell_2/\ell_1\cap\ell_2}d\mu_{\ell_2/\ell_2\cap\ell_1}\\
   &\ \ \ \ \ \ \ \ \ \ \ \ \ \ \ \ \ \alpha\chi(\frac{1}{2}J(u,z)+\frac{1}{2}J(v,u+z)+\frac{1}{2}J(w,u+v+z))\psi(u+v+w+z)\\
   =&\int_{\ell_1/\ell_1\cap\ell_3}d\mu_{\ell_1/\ell_1\cap\ell_3}\int_{\ell_3/\ell_2\cap\ell_3}d\mu_{\ell_3/\ell_2\cap\ell_3}\int_{\ell_2/\ell_1\cap\ell_2}d\mu_{\ell_2/\ell_2\cap\ell_1}\\
  &\ \ \ \ \ \ \ \ \ \ \ \ \ \ \ \ \  \alpha\chi(\frac{1}{2}J(w,v)+J(v+w,u)+\frac{1}{2}J(v+w,z))\psi(v+w+z)\\
   =&(\int_{v\in\ell_3/\ell_2\cap\ell_3,w\in\ell_2/\ell_1\cap\ell_2,v+w\in\ell_1}\chi(\frac{1}{2}J(w,v))d\mu_{\ell_3/\ell_2\cap\ell_3}d\mu_{\ell_2/\ell_2\cap\ell_1})\psi(g),
   \end{align*}
   where  $\psi$ has been viewed as a function on $V$ in the first and the second equalities. The properties of Maslov index can be easily deduced from the definition\cite{z}.
\end{proof}
\begin{example}Let us consider an  example where $\ell_1,\ell_2, \ell_3\in Lag(1,V)$. There are two non-trivial  cases.

(i) If $(\ell_1,\ell_2, \ell_3)=(\ell_a,\ell_b, \ell_c), a\neq b\neq c$, we can directly calculate
\begin{align*}
  \alpha(\ell_a,\ell_b, \ell_c)&=\int_{v\in\ell_c,w\in\ell_b,v+w\in\ell_a}\chi(\frac{1}{2}J(w,v))d\mu_{\ell_c}d\mu_{\ell_b}\\
  &=\int_{\mathbb{Q}_p}\chi(\frac{(c-b)(a-c)}{2(b-a)}x^2)dx/|\int_{\mathbb{Q}_p}\chi(\frac{(c-b)(a-c)}{2(b-a)}x^2)dx|\\
  &=\left\{
      \begin{array}{ll}
        \lambda_p(\frac{(c-b)(a-c)}{2(b-a)}), & \hbox{$p>3$;} \\
       \frac{1}{2} \lambda_2(\frac{(c-b)(a-c)}{2(b-a)}), & \hbox{$p=2$.}
      \end{array}
    \right.
\end{align*}

(ii) If $(\ell_1,\ell_2, \ell_3)=(\ell_a,\ell_b, \ell_*), a\neq b$, we similarly have
\begin{align*}
  \alpha(\ell_a,\ell_b, \ell_*)
  &=\int_{\mathbb{Q}_p}\chi(\frac{a-b}{2}x^2)dx/|\int_{\mathbb{Q}_p}\chi(\frac{a-b}{2}x^2)dx|\\
  &=\left\{
      \begin{array}{ll}
        \lambda_p(\frac{a-b}{2}), & \hbox{$p>3$;} \\
       \frac{1}{2} \lambda_2(\frac{a-b}{2}), & \hbox{$p=2$.}
      \end{array}
    \right.
\end{align*}

\end{example}

\end{document}